\documentclass[11pt]{article}
\usepackage[utf8]{inputenc}
\usepackage[english]{babel}
\usepackage[T1]{fontenc}
\usepackage[margin=1.15in]{geometry}
\usepackage{amsmath}
\usepackage{amsthm}
\usepackage{amsfonts}
\usepackage{graphicx}
\usepackage{float}
\usepackage{pgfplots}
\pgfplotsset{compat=1.14, set layers}
\usepackage{bm}

\usepackage{mathtools}
\usepackage{footmisc}
\usepackage{amssymb}
\usepackage{subfig}
\usepackage{enumitem}
\usepackage{setspace}
\usepackage{caption} 
\usepackage{float}
\usepackage{thm-restate}
\usepackage{xspace,xcolor}
\usepackage{thmtools}
\usepackage{framed}
\usepackage{soul}
\usepackage{algorithm}
\usepackage[noend]{algpseudocode}

\definecolor{refcolor}{rgb}{0.23, 0.27, 0.29}
\usepackage[breaklinks,colorlinks,urlcolor={refcolor},citecolor={refcolor},linkcolor={refcolor}]{hyperref}
\usepackage[capitalize, nameinlink]{cleveref}

\newtheorem{theorem}{Theorem}[section]
\newtheorem*{theorem*}{Theorem}
\newtheorem{lemma}[theorem]{Lemma}

\newtheorem{corollary}[theorem]{Corollary}
\newtheorem{definition}[theorem]{Definition}

\newtheorem{claim}[theorem]{Claim}
\newtheorem{observation}[theorem]{Observation}

\newtheorem{restatement}[theorem]{Restatement}
\newtheorem*{fact*}{Fact}

\newtheorem*{lemma*}{Lemma}
\newtheorem*{definition*}{Definition}

\Crefname{paragraph}{Paragraph}{Paragraphs}
\Crefname{observation}{Observation}{Observations}
\Crefname{restatement}{Restatement}{Restatement}

\newcommand{\mpc}{\textsf{MPC}\xspace}
\newcommand{\local}{\textsf{LOCAL}\xspace}

\newcommand{\lcl}{\textsf{LCL}\xspace}

\newcommand{\vpivot}{V^{pivot}\xspace}  % !!!need to propagate this macro

\newcommand{\diam}{\text{diam}}

\newcommand{\arr}{\rightarrow}
\newcommand{\narr}{\not\rightarrow}
\newcommand{\larr}{\leftarrow}

\newcommand{\be}{\sf BalancedExponentiation\xspace}
\newcommand{\pro}{{\sf ProbeDirections}\xspace}
\newcommand{\expo}{{\sf Exp}\xspace}
\newcommand{\hds}{{\sf blockedDirs}\xspace}

\newcommand{\act}{\textsf{active}\xspace}
\newcommand{\full}{\textsf{full}\xspace}
\newcommand{\comp}{\textsf{knowledgeable}\xspace}
\newcommand{\blocked}{\textsf{blocked}\xspace}

\newcommand{\layer}{\operatorname{layer}}

\DeclareMathOperator*{\argmax}{arg\,max}

\newcommand{\poly}{\operatorname{\text{{\rm poly}}}}
\newcommand{\logstar}[1]{\log^{*} #1}

\newif\ifdraft
\drafttrue

\newcommand{\eps}{\varepsilon}

\begin{document}

\begin{center}
\begin{minipage}[H]{14.5cm} 

\begin{center}
{\huge \bf Conditionally Optimal Parallel Coloring of Forests}
\end{center}
\vspace{1cm}

{\large \textbf{Christoph Grunau}, ETH Zürich -- \href{mailto:cgrunau@inf.ethz.ch}{\texttt{cgrunau@inf.ethz.ch}}} \vspace{1mm}\\
{\large \textbf{Rustam Latypov\footnotemark[1]}, Aalto University -- \href{mailto:rustam.latypov@aalto.fi}{\texttt{rustam.latypov@aalto.fi}}} \vspace{1mm}\\
{\large \textbf{Yannic Maus\footnotemark[2]}, TU Graz -- \href{mailto:yannic.maus@ist.tugraz.at}{\texttt{yannic.maus@ist.tugraz.at}}} \vspace{1mm}\\
{\large \textbf{Shreyas Pai\footnotemark{}}, Aalto University -- \href{mailto:shreyas.pai@aalto.fi}{\texttt{shreyas.pai@aalto.fi}}} \vspace{1mm}\\
{\large \textbf{Jara Uitto}, Aalto University -- \href{mailto:jara.uitto@aalto.fi}{\texttt{jara.uitto@aalto.fi}}} \vspace{1mm}\\

\vspace{5mm}

\begin{center}
	{\bf Abstract} \\ 
\end{center}

We show the first conditionally optimal deterministic algorithm for  $3$-coloring forests in the low-space massively parallel computation (\mpc) model. Our algorithm runs in $O(\log \log n)$ rounds and uses optimal global space. The best previous algorithm requires $4$ colors [Ghaffari, Grunau, Jin, DISC'20] and is randomized, while  our algorithm are inherently deterministic. \\

Our main technical contribution is an $O(\log\log n)$-round algorithm to compute a partition of the forest into  $O(\log n)$ ordered layers such that every node has at most two neighbors in the same or higher layers. Similar decompositions are often used in the area and we believe that this result is of independent interest. Our results also immediately yield conditionally optimal deterministic algorithms for maximal independent set and maximal matching for forests, matching the state of the art [Giliberti, Fischer, Grunau, SPAA'23]. In contrast to their solution, our algorithms are not based on derandomization, and are arguably simpler.
\end{minipage}

\thispagestyle{empty}
\footnotetext{Supported by the Academy of Finland, Grant 334238}
\footnotetext[2]{Supported by the Austrian Science Fund (FWF), Grant P36280-N}

\end{center}

\newpage
\thispagestyle{empty}
\tableofcontents

\newpage
\pagenumbering{arabic}

\section{Introduction}

A recent sequence of papers investigates fundamental symmetry-breaking problems such as coloring, maximal independent set and maximal matching on trees \cite{BRANDT202122,behnezhadMIS,LU21,GGJ20,GFG23}.
We conclude, simplify and unify this line of work by giving a conceptually simple algorithm for $3$-coloring, maximal independent set and maximal matching. We solve the three problems in a unified way by computing a so-called $H$-decomposition (we discuss these in more detail in \Cref{sec:contributions}). Even though such decompositions are the natural tool for solving the aforementioned problems on trees, computing them efficiently in the \mpc model remained outside the reach of previous techniques.

\begin{restatable}{theorem}{thmMainColoringMISMatching}
	There are deterministic $O(\log\log n)$-round low-space \mpc algorithms for  $3$-coloring, maximal matching and maximal independent set (MIS) on forests. These algorithms use $O(n)$ global space.
\end{restatable}

The runtimes of our algorithms are conditionally optimal, conditioned on the 1 vs 2 cycle conjecture, at least if one restricts to so-called component stable algorithms \cite{Ghaffari2019, componentstable,Linial92,Roughgarden18} (see \Cref{sec:relatedWork} for a brief discussion about component-stability).  

We note that algorithms for maximal matching and maximal independent set matching our guarantees are known from a very recent work \cite{GFG23}. However, their algorithms are quite complicated and technical, and use sophisticated derandomization techniques. Moreover, their techniques inherently cannot be used to color a tree with a small number of colors. Indeed, the $3$-coloring problem is considered to be the hardest of the three problems, e.g., once such a coloring is known one can compute an MIS in $O(1)$ rounds. Additionally, a crucial property used in previous \mpc algorithms for MIS and maximal matching is that any partial solution can be extended to a solution of the whole graph; a property that does not hold for $3$-coloring. 
The best previous algorithm for coloring trees uses $4$ colors and is randomized \cite{GGJ20}. If one allows for randomization the single additional color makes the problem significantly easier by the following divide and conquer approach: if one partitions the tree into two parts by letting each node join one of the parts uniformly at random, the connected components induced by each part have logarithmic diameter. Once the diameter is small, one can use $O(\log\log n)$ \mpc rounds to color each component independently with two colors in a brute force manner. 
In the next section we zoom out and present the bigger picture of our work.

\subsection{\mpc Model and Exponential Speed-Up Over \local Algorithms}

The Massively Parallel Computation (\mpc) model~\cite{KarloffSV10} is a mathematical abstraction of modern frameworks of parallel computing such as Hadoop~\cite{White2009}, Spark~\cite{Zaharia2010}, MapReduce~\cite{Dean2008}, and Dryad~\cite{Isard2007}. 
In the \mpc model, we have $M$ machines that communicate in all-to-all fashion, in synchronous rounds.
In each round, every machine receives the messages sent in the previous round, performs (arbitrary) local computations, and is allowed to send messages to any other machine.
Initially, an input graph of $n$ nodes and $m$ edges is arbitrarily distributed among the machines.
At the end of the computation, each machine needs to know the output of each node it holds, e.g., their color in the vertex-coloring problem.

The \mpc model is typically divided into $3$ regimes according to the local space $S$. The \emph{superlinear} and the \emph{linear} regimes allow for $S = n^{1 + \Omega(1)}$ and $S = \widetilde{O}(n)$ words\footnote{The $\widetilde{O}$ notation hides polylogarithmic factors.} of space (memory) per machine. 
A word is $O(\log n)$ bits and is enough to store a node or a machine identifier from a polynomial (in $n$) domain. The local space restricts the amount of data a machine initially holds and is allowed to send and receive per round. Both linear and superlinear regimes allow for very efficient algorithms because machines can get a ``global view'' of the graph in the sense that it can store information for each node of the graph \cite{filtering2011, linear-MIS, Ghaffari2019-arboricity}. 
However, the growing size of most real-world graphs makes it impossible to get such a global view on a single machine and hence research in recent years has focused on the most challenging \emph{low-space} (or \emph{sublinear}) regime with $S = n^{\delta}$, for some constant $\delta < 1$, where we cannot even store the whole neighborhood of a single node in a single machine. As each machine can only get a local view there are close connections to the \local model of distributed computing that we further elaborate on below.

Furthermore, we focus on the most restricted case of \emph{linear global space}, i.e., $S \cdot M = \Theta(n + m)$.
Notice that $\Omega(n + m)$ words are \emph{required} to store the input graph.

\subparagraph*{The \local Model and Graph Exponentiation.} The \local model is a classic model of distributed message passing. Each node of an input graph hosts a processor and the nodes communicate along the edges of the graph in synchronous rounds. The local computation, local space, and message sizes are unbounded in this model. 
Most research in the \local model has focused on symmetry breaking problems like graph colorings, MIS, and maximal matchings. For most of these classic problems $O(\log n)$-round randomized algorithms are known \cite{luby86,alon86,coloring-simple} which can directly be translated to the \mpc model.
A major focus on recent and current research is to develop \emph{sublogarithmic} \mpc algorithms that beat the logarithmic baseline.

In fact, the strong connection between the models also shows up in faster algorithms, as almost all recent \mpc algorithms for such problems are \mpc-optimized implementations of algorithms that were originally developed for the \local model.
The main technique to obtain this speedup is the graph \emph{exponentiation technique}~\cite{wattenhofer}. It allows to gather the $T$-hop radius neighborhood of a node in $O(\log T)$ \mpc rounds. 
So, as long as these neighborhoods fit the space constraints, after gathering them one can simulate $T$-round \local algorithms locally to compute the output for each node.
Furthermore, for component-stable algorithms, the connection also goes the other way around, that is, an $\Omega(T)$ lower bound on the round complexity in the \local model implies an exponentially lower $\Omega(\log T)$ conditional lower bound in the \mpc model.  
Thus, the holy grail is to obtain this exponential speedup over the \local model.
A central open problem in the area is to find an $O(\log\log n)$ round \mpc algorithm for the classic MIS problem on general graphs, which enjoys a matching conditional $\Omega(\log \log n)$-round lower bound. 

Unfortunately, we are very far from answering this question.
The current state of the art is  an $\widetilde{O}(\sqrt{\log \Delta} + \log \log \log n)$-round randomized \mpc algorithm~\cite{GU19}. 
We note that we take into account the new results on network decomposition, which reduce the dependency on $n$~\cite{RG20,Ghaffari2023-soda}.
This result is obtained by combining the graph exponentiation technique with \emph{sparsification} methods~\cite{wattenhofer, GU19}.
The exponentiation technique is used to simulate Ghaffari's ${O}(\log \Delta + \poly \log \log n)$-round MIS algorithm for the \local model~\cite{GhaffariImproved16}.

From a high-level perspective they break the \local algorithm into $O(\sqrt{\log \Delta})$ phases each of length $T=O(\sqrt{\log \Delta})$.
In the beginning of each phase, the graph is subsampled so that the maximum degree of any node is at most $2^{\sqrt{\log \Delta}}$.
Then, we can gather the $T$-hop neighborhood of each node in $O(\log \log \Delta)$ \mpc rounds and simulate $T$ rounds of the \local algorithm in a single \mpc round. The main benefit of  simulating (shorter) phases and subsampling to smaller degree graphs is to reduce the memory resources needed during the exponentiation technique. Unfortunately, this phase-based approach seems to hit a fundamental barrier at  $\sqrt{\log \Delta}$ rounds, and it is unclear how to reduce memory usage without it.  Due to little progress in improving on this result, recent research has focused on special graph classes such as trees and bounded arboricity graphs. 

\subparagraph*{Symmetry Breaking on Trees and Bounded Arboricity Graphs.}

Studying low-space \mpc algorithms for MIS on trees and forests has been fruitful. This line of work started with a randomized $O(\log^3 \log n)$ round low-space \mpc algorithm for MIS and maximal matching on trees \cite{BRANDT202122}. 
Later, the round complexity was first improved to $O(\log^2 \log n)$ \cite{behnezhadMIS} and finally to $O(\log \log n)$ \cite{GGJ20}, where both algorithms extend to low-arboricity graphs. The $O(\log \log n)$ algorithm is conditionally optimal, at least if one restricts oneself to component-stable algorithms. Finally, a recent work derandomized the $O(\log \log n)$ round algorithm using \mpc specific derandomization techniques and thus obtained a deterministic $O(\log \log n)$ round MIS and Maximal Matching algorithm for trees and more generally low-arboricity graphs \cite{GFG23}.

While especially the $O(\log \log n)$ round algorithms are quite technical and involved, all of the aforementioned previous algorithms rely on the same fundamental idea. 
Namely, to interleave graph exponentiation with the computation of partial solutions to rapidly decrease the maximum degree of the remaining graph.
Unfortunately, it seems unlikely that such a rapid degree reduction is possible in general graphs; thus it seems that new approaches are necessary in order to get an $O(\log \log n)$ round algorithm for general graphs.

Also, their approach does not work for coloring a forest with a constant number of colors. The main reason is that they critically rely on the fact that any partial solution can be extended to a full solution, which is not the case for coloring a forest with a fixed number of colors.

\subsection{Our Technical Contribution}
\label{sec:contributions}

We present a unified solution for $3$-coloring, MIS, maximal matching that takes $O(\log \log n)$ rounds. The core technical contribution that unifies these is an efficient algorithm to compute \emph{$H$-decompositions}. 

\begin{restatable}{theorem}{thmStrictHdecomp}
	There is a deterministic $O(\log \log n)$-rounds low-space \mpc algorithm that computes a strict $H$-decomposition with $O(\log n)$ layers on forests in $O(n)$ global space. 
\end{restatable}

$H$-decompositions were introduced to the area of distributed computing by Barenboim and Elkin \cite{BE10}. An $H$-decomposition (of a forest) partitions the vertices of the graph into layers such that every node has at most two neighbors\footnote{There are generalizations to higher number of neighbors that are important when dealing with bounded arboricity graphs~\cite{NashWilliams1964, BEPSv3}.} in higher or equal layers. For forests an $H$-partition with $O(\log n)$ layers always exist. 
In the \local model, such an $H$-decomposition immediately implies an algorithm for $3$-coloring in $O(\log n)$ rounds. Essentially, one can iterate through the layers in a reverse order and color all nodes in a layer while avoiding conflicts with the already colored neighbors in higher layers. 

The novelty of our approach is not that we use such decompositions to compute a $3$-coloring of a forest, in fact, this straightforward approach has made it into the classrooms of many graduate programs of universities, but in the way how we compute it. We detail on our solution in more detail in the nutshell, but the main take-away is as follows. We steer every machine to learn some parts of the graph (to large extent in an uncoordinated fashion) such that every machine can compute a partial $H$-decomposition locally, which  we can later unify to a global decomposition. We are not aware of any other \mpc algorithm for $H$-decompositions with a similar approach.

\subparagraph*{Balanced Exponentiation.}
In order to achieve exponential speedup, our algorithms rely on graph exponentiation. However, there are no known sparsification techniques that can cope with the memory resources that are needed for the classic graph exponentiation technique. 
Instead, we provide a self-contained exponentiation procedure whose memory overhead is very mild on forests. To explain our procedure, we first need to define a \emph{subtree}. A subtree is a subgraph of a tree, such that if it is removed, the rest of the tree stays connected. A node is \emph{important} if it is contained in a subtree of size $n^{\delta/8}$. We present the following result (the formal statement appears in \Cref{lem:balancedExp}). 

\medskip

\noindent \emph{Let $0 < k \leq n^{\delta/8}$  be a parameter. There is a deterministic low-space \mpc algorithm that, given an $n$-node forest $F$, uses $O(\log k)$ rounds in which every important node $v\in F$ discovers its $k$-hop neighborhood in every direction of the graph, except for at most one.}

\medskip

Given a node $v\in F$, we refer to each of its neighbors $x\in N(v)$ as a \emph{direction} with regard to $v$. Informally, what node $v$ can discover in direction $x$ is simply the subgraph of $F$ that is connected to $v$ via $x$, which is uniquely defined, since $F$ is a forest.

This result above may be of independent interest and may be useful to design algorithms for other graph problems. 
We obtain it by extending the exponentiation technique of  a recent work by \cite{sodapaper}. Their work designs an exponentiation technique which (almost) equals ours in the special case when $k$ equals the maximum diameter of a component of the forest. In their work it is used in an $O(\log \diam)$-round algorithm to compute the connected components of a forest. Later it has also been used to solve certain dynamic programming tasks on tree-structured data \cite{gupta2023fast}, also in time that is logarithmic in the diameter. In \Cref{sec:algoOverview} we present a more detailed discussion on the similarities and the difference between the exponentiation result in this work and the one in their work, and why our result requires a different analysis. The main benefit of our result is that a flexible choice of $k$ allows the runtime and space to be small, if the required ``view'' for the nodes is small, which we heavily utilize in our algorithm to compute $H$-decompositions.

\subsection{Our Method in a Nutshell}
\label{sec:nutshell}

As mentioned in the previous section, our key technical contribution is to compute a so-called $H$-decomposition of the input forest $F$.
In particular, the goal is to compute a partition $V(F) = V_1 \sqcup V_2 \sqcup \ldots \sqcup V_L$ of the vertices into $L = O(\log n)$ layers such that each node in $V_i$ has at most two neighbors in $\bigcup_{j \geq i}V_j$.
There exists a simple peeling algorithm which computes such a partition; iteratively peel off all nodes of degree at most 2 and define $V_i$ as the set of nodes that got peeled off  in the $i$-th iteration.
A simple calculation shows that at least half of all the remaining nodes get peeled off in each iteration, and hence we get a decomposition into $O(\log n)$ layers.
Moreover, one can determine the iteration in which a node gets peeled off by only looking at its $O(\log n)$-hop neighborhood.
Thus, if we could compute for a given node its entire $O(\log n)$-neighborhood and store it in a single machine, then we could locally determine the layer of that node with no further communication.

One way to compute the $O(\log n)$-neighborhood of each node in the \mpc model is the well-known graph exponentiation technique.
Generally speaking, graph exponentiation allows to learn the $2^i$-hop neighborhood of each node in $O(i)$ \mpc rounds.
Thus, we could in principle hope to learn the $O(\log n)$-hop neighborhood of each node in just $O(\log \log n)$ rounds.
However, one obviously necessary precondition of the graph exponentiation technique is that the $O(\log n)$-hop neighborhood of each node has size $n^{\delta}$, as otherwise we cannot possibly store the neighborhood in one machine.
This is quite a limiting condition. If the input is for example a star, even the two-hop neighborhood of each node contains $\Omega(n)$ vertices. Moreover, even if each local neighborhood would fit into one machine, the global space required to store all the neighborhoods might still be prohibitively large, especially if one aims for near-linear global space.

Thus, we cannot use the vanilla graph exponentiation technique. Instead, we use the balanced graph exponentiation technique for forests mentioned in the previous section.
The output guarantee of the balanced exponentiation algorithm, running in $O(\log \log n)$ rounds, weakens the guarantee that each node sees its $O(\log n)$-hop in two ways.
First, it only gives a guarantee for nodes that are contained in a sufficiently small subtree, namely of size at most $n^{\delta}$. Second, for each node $v$ in a small subtree, it computes all nodes of distance $O(\log n)$, except for nodes in one direction. 

We start by briefly discussing how one can deal with the first shortcoming. If one iteratively removes all nodes that are contained in a subtree of size at most $x$ from $F$ and all nodes of degree at most $2$, then all nodes are removed within $O(\log_x(n))$ iterations. This fact was used in similar forms in previous results and for completeness we give a standalone proof (see \cref{lem:generalizedRakeCompress}).
Thus, if we repeatedly assign nodes in subtrees of size at most $n^{\delta}$ and nodes of degree at most $2$ to one of $O(\log n)$ layers, then after $O(1/\delta)$ iterations, we assigned each node to one of $O((1/\delta) \log n)$ layers.
Thus, it intuitively suffices to focus on nodes in subtrees of size at most $n^{\delta}$. \cref{sec:strict_h} gives a formal treatment of this argument.

The more severe difficulty stems from the fact that there might not be a single node in the forest for which we have stored its entire $O(\log n)$-neighborhood in one machine.
This makes it impossible to locally determine the layer of each node, or even a single one, assigned by the simple peeling process described in the beginning.
Instead, each node $v$ locally simulates a conservative variant of the peeling algorithm described above; in each iteration not all the nodes of degree at most $2$ are removed, but only those that $v$ has stored in its machine.
Note that if $v$ has strictly more than $2$ neighbors not stored in its machine, then the conservative peeling algorithm would never peel off $v$.
Moreover, even if $v$ would eventually be peeled off, then there is no guarantee that it happens within the first $O(\log n)$ iterations.
However, the fact that $v$ has stored all the nodes in its $O(\log n)$-hop neighborhood except for nodes in one direction in its machine suffices to show that $v$ gets peeled off within the first $O(\log n)$ iterations (see \cref{lem:conservative_progress}).
Thus, each node $v$ locally computes a layering $V^{v} = V^{v}_1 \sqcup  \dots V^{v}_L$ for some $L = O(\log n)$ such that $v \in V^v$ and each node in $V^{v}_i$ has at most two neighbors contained in $\left(\bigcup_{j \ge i} V^v_j \right) \cup \left(V(F) \setminus V^v \right)$. As some nodes might not be assigned to any layer, we refer to such a decomposition as a partial $H$-decomposition.
Note that a node might get assigned to different layers from different nodes. Fortunately, this is not a problem because of the following nice structural property about (partial) $H$-decompositions:
if we are given multiple (partial) $H$-decompositions, then we can get another (partial) $H$-decomposition by assigning each node to the smallest layer assigned by any of the $H$-decompositions.
This structural observation allows us to combine the different locally computed (partial) $H$-decompositions into a single partial $H$-decomposition where each node in a small subtree is assigned to one of the $O(\log n)$ layers. 

\subparagraph*{Rooted vs. Unrooted Forests.} Our results are for unrooted forests, which are indeed more difficult than rooted forests. In fact, the fastest known \mpc algorithm to root a forest takes $O(\log \diam)$ rounds (and at least on general forests this runtime is conditionally tight) \cite{sodapaper}; so rooting the forest does not fit our time budget of $O(\log \log n)$ rounds. Many steps of our algorithm would simplify (or maybe even allow for alternative solutions) if the forest was rooted. For example, in a directed forests, we would not need our balanced exponentiation procedure. One can show that nodes can just exponentiate towards their children until the local memory is full without breaking any global memory bounds. However we would still need our combinatorial algorithm for creating a (global) $H$-partition. Observe that \cite{sodapaper} also contains a $O(\log \diam)$-round 2-coloring algorithm for rooted constant-degree forests, which can be generalized to rooted unbounded-degree forests. 

\subsection{Further Related Work}
\label{sec:relatedWork}
\subparagraph*{Component Stability.}
Roughly speaking, an \mpc algorithm is component-stable, if the outputs of nodes in different components are independent of each other.
Low-space component-stable \mpc algorithms are closely connected to algorithms in the \local model and this connection was used to lift (unconditional) lower bounds from the \local model into conditional lower bounds in the \mpc model~\cite{Ghaffari2019}.
Under the 1 vs 2 cycle conjecture, this technique turns an $\Omega(T)$-round lower bound in \local into an $\Omega(\log T)$ lower bound in low-space \mpc.
This approach was used to establish, among others, $\Omega(\log \log n)$ randomized lower bounds for MIS and maximal matching.
Later, the technique was extended to deterministic component-stable algorithms as well~\cite{componentstable}.
While the assumption of component-stability might seem very natural to \mpc algorithms, it is known that component-instability can help.
For example, any component-stable algorithm for finding an independent set of size $\Omega(n / \Delta)$ requires $\Omega(\log \log^* n)$ rounds, while there is an $O(1)$-round algorithm that is not component-stable~\cite{componentstable}.

\subparagraph*{log$(\diam)$ Algorithms on Forests.} There are surprisingly few works with a strict $\log(\diam)$ runtime for any graph families in any \mpc regimes, where $\diam$ refers to the diameter. To our knowledge, the only existing ones are low-space algorithms for forests \cite{sodapaper,gupta2023fast}. The authors of \cite{sodapaper} show that connectivity, rooting, and all \lcl (locally checkable labeling) problems can be solved on forests in $O(\log \diam)$ using optimal global space $O(n)$. The authors of \cite{gupta2023fast} build on top of the works of \cite{sodapaper} by introducing a framework to solve dynamic programming tasks and optimization problems, all in time $\log(\diam)$ and global space $O(n)$.
We note that given a double-logarithmic dependency on $n$, the connectivity problem can be solved on general graphs (even deterministically) in $O(\log \diam + \log \log n)$ time using linear total space~\cite{Behnezhad2019, ccderandom}.
Also, the 1 vs 2 cycle conjecture directly rules out an $o(\log \diam)$ for connectivity.

\subparagraph*{Symmetry-Breaking on General Graphs.}
In general graphs, $(\Delta + 1)$-vertex coloring is an intensively studied symmetry breaking problem, where $\Delta$ is the maximum degree of the graph.  A series of works \cite{P18,BKM20,PS18,Chang2019} for  Congested Clique model which is similar to the \mpc model with linear local memory has culminated in a deterministic $O(1)$-round algorithm~\cite{CDP21}. 

In the low-space \mpc model, the first algorithm for the problem was randomized and used  $O(\log \log \log n)$ rounds with almost linear $\widetilde{O}(m)$ global space~\cite{Chang2019}\footnote{The $O(\sqrt{\log \log n})$ runtime stated in the paper is automatically improved to $O(\log \log \log n)$ through developments in network decomposition~\cite{RG20}.}.
By derandomizing the classic logarithmic-time algorithms, one can obtain an $O(\log \Delta + \log \log n)$-round algorithm for $(\Delta + 1)$-coloring, MIS, and maximal matching~\cite{Czumaj2020, Czumaj-congested-coloring}.
For coloring, this was improved to $O(\log \log \log n)$ through derandomizing a tailor-made algorithm~\cite{Czumaj2021}.
The deterministic algorithms require $n^{1 + \Omega(1)}$ global space.

\subsection{Outline}
We define strict $H$-decompositions in \Cref{sec:StrictH}, and then we show how to compute them in $O(\log \log n)$ low-space \mpc rounds using $O(n\cdot \poly(\log n))$ global space in \Cref{sec:strict_h}. The key subroutine for the algorithm in \Cref{sec:strict_h} is discussed in \Cref{sec:subtreeRC}.  In \Cref{sec:coloring} we show how to use these decompositions to compute a coloring, MIS, and matching. In \Cref{sec:optimalSpace} we show how to reduce the global memory usage of our algorithms from $O(n\cdot \poly(\log n))$ to $O(n)$. The balanced exponentiation procedure appears in \Cref{sec:balanced_exponentiation}.
%Missing proofs of Section~\ref{sec:strict_h} appear in \Cref{app:proofOfWrapper}.

\section{Preliminaries and Notation}

The input graph is an undirected, finite, simple forest $F = (V,E)$ with $n=|V|$ nodes and $m=|E|$ edges such that $E \subseteq [V]^2$ and $V \cap E = \emptyset$. For a subset $S \subseteq V$, we use $G[S]$ to denote the subgraph of $G$ induced by nodes in $S$. 

Let $\deg_F(v)$ denote the degree of a node $v$ in $F$ and let $\Delta$ denote the maximum degree of $F$. For node set $S \subseteq V(F)$ and a node $v\in S$ we write $\deg_S(v)$ for the degree of $v$ in $F[S]$. The distance $d_F(v,u)$ between two vertices $v,u$ in $F$ is the number of edges in the shortest $v - u$ path in $F$; if no such path exists, we set $d_F(v, u) \coloneqq \infty$. Sometimes we simply write $\deg(v)$ and $d(v, u)$ if it is clear from context that we refer to the degree and distance in graph $F$. The greatest distance between any two vertices in $F$ is the diameter of $F$, denoted by $\diam(F)$. 

For each node $v$ and for every $k \in \mathbb{N}$, we denote the $k$-hop (or $k$-radius) neighborhood of $v$ as $N^k(v) = \{ u \in V : d(v,u) \leq k\}$. Set $N^1(v)$ is simply the set of neighbors of $v$, to which often refer to as $N(v)$. We often consider sets of nodes $S$ from which we need to remove a single node $u$. Hence, we use the notation $S\setminus u$ as a shorthand for $S\setminus \{u\}$.

\section{Strict \texorpdfstring{$H$}{Lg}-decompositions}
\label{sec:StrictH}
We begin with the formal definition of an $H$-decomposition, that is, a partition of the graph into layers such that every node has at most two neighbors in higher or equal layers. We also extend the definition to the setting where some nodes remain without a layer. 
\begin{definition}[(Partial) $H$-Decomposition] 
	\label{def:Hdecomp}
	Let $F$ be a forest and $\layer \colon V(F) \mapsto \mathbb{N} \cup \{\infty\}$. For $i\in \mathbb{N} \cup \{\infty\}$ define  $V_i=\{v\in V(F) \mid \layer(v)=i\}$, $V_{\geq i}=\bigcup_{j\geq i} V_j$. 
	
	We say that $\layer$ is a \emph{partial $H$-decomposition} if $deg_{V_{\geq i}}(v)\leq 2$ holds for every $v$ with $\layer(v)=i$. 
	We speak of an \emph{$H$-decomposition} if $V_{\infty}=\emptyset$, and $L=\max\{\layer(v) \mid v\in V(F), \layer(v)\neq \infty\}$ is the \emph{length} of the decomposition. 
\end{definition}
We also refer to the $V_i$'s as the \emph{layers} of the (partial) $H$-decomposition. 

\subparagraph*{Why $3$-coloring and strict $H$-decompositions?} $H$-decompositions were introduced to the area of distributed computing by Barenboim and Elkin \cite{BE10}. Nowadays, they are a frequent tool in the area and by increasing the degree bound $2$ to $(2+\eps)a$ the concept also extends to graphs with arboricity at most $a$ (this is the original setting considered in \cite{BE10}). More generally, it can be shown that $H$-decompositions with $O(\log n)$ layers exist. In the \local model, an $H$-decomposition of a tree with $O(\log n)$ layers can be computed by iteratively removing nodes of degree $1$ (rake) and nodes of degree $2$ (compress). 

In the \local model, one can $3$-color any graph with a given $H$-decomposition as in \Cref{def:Hdecomp} with $L$ layers in $O(L +\logstar n)$ rounds.  Each layer induces a graph with maximum degree $2$. First, use Linial's algorithm to color each layer in parallel with $C=O(1)$ colors. This coloring may contain lots of monochromatic edges between different layers and is only used as a schedule to compute the final $3$-coloring. In order to compute that final coloring, iterate through the layers in a decreasing order, and in each layer iterate through the $C$ colors. When processing one of the $C$ color classes, every node picks one color in $\{1,2,3\}$ not used by any of its already colored neighbors (at most two). 

\smallskip

Optimally, we would like to use the above \local model algorithm as the base for our exponentially faster \mpc algorithm. However, even if we were given an $H$-decomposition for free it is non-trivial to actually use it for $3$-coloring a graph if the runtime is restricted to $O(\log\log n)$ rounds and you only allow for a polylogarithmic memory overhead. Going through the layers in some sequential manner would be way too slow as it would require logarithmically rounds. Still, as the \local algorithm has locality $T=\Theta(\log n)$ the output of a node may depend on the topology in logarithmic distance. In order to achieve a fast \mpc algorithm, we want the nodes  to use the graph exponentiation technique to learn the part $G_v$ of their $T$-hop neighborhood that is relevant to determine their output in $O(\log T)=O(\log \log n)$ rounds. We refer to $G_v$ as the predecessor graph of node $v$. The challenge with the standard $H$-decomposition as given by \Cref{def:Hdecomp} is that, even though $G_v\subseteq V_{\geq \layer(v)}$, it may be  of size $\Theta(n)$. Hence, even if every node could learn its predecessor graph and store it in its local memory $S_v$ (formally, the memory of every node is stored on some machine), the global space bound can only be upper bounded by $\sum_{v\in V}|G_v|=O(n^2)$, drastically, violating the desired near-linear bound. In order to circumvent this issue we introduce the concept of a strict $H$-decomposition which is optimized for its usage in the \mpc model. The bottom line of this decomposition is that besides the properties of a classic $H$-decomposition, we also have a set $\vpivot$. The set $\vpivot$ induces a graph with maximum degree $2$ and hence can be colored with $3$ colors in $O(\logstar n)$ rounds with Linial's algorithm \cite{Linial92}. The main gain compared to the classic $H$-decomposition is, that once we have colored the nodes in $\vpivot$, we can show that the predecessor graph of every node $v\in V\setminus \vpivot$ is of logarithmic size (when considering the same \local model algorithm that colors these nodes layer by layer). Hence, each node can learn its predecessor graph in $O(\log\log n)$ rounds without violating global space constraints. We next present the definition of a strict $H$-decomposition. 

\begin{definition}[(Partial) Strict $H$-Decomposition]
	\label{def:strictH}
	Let $F$ be a forest and $\layer \colon V(F) \mapsto \mathbb{N} \cup \{\infty\}$ be a function. We define $V_{<\infty} = \{v \in V(F) \mid \layer(v) < \infty\}$ and
	\begin{align*}
		\vpivot := \{v \in V_{<\infty} \mid \layer(v) \geq \layer(w) \text{ for every $w \in N_F(v)$}\}.
	\end{align*}
	
	We refer to a (partial) $H$-decomposition $\layer$ as \emph{strict} if for every $v \in V_{<\infty} \setminus \vpivot$, it holds that
	\begin{align}
		\label{cond:pivot}
		|\{w \in N_F(v) \setminus \vpivot \mid \layer(w) = \layer(v)\} \cup \{w \in N_F(v) \mid \layer(w) > \layer(v)\}| \leq 1.
	\end{align}
	
	That is, the total number of non-pivot neighbors with the same layer and of neighbors with a strictly higher layer is at most $1$.
\end{definition}
There is some similarity between the definition of a strict $H$-decomposition and the $H$-decompositions used in the theory of so called locally checkable labelings \cite{CP19,Chang2020,BCMOS21}. These decompositions iteratively layer degree $1$ nodes and paths of length at least $\ell$. Our strict $H$-decomposition is similar to the case when we remove paths of length at least $\ell = 3$ (see \Cref{lem:strictH-generalizedrc-connect}).
The following lemma is one of the most crucial structural properties of partial $H$-decompositions that we exploit in the core of our algorithm (see \Cref{sec:subtreeRC}).

\begin{lemma}[Partial Strict $H$-Decomposition, Closure under taking minimums]
	\label{lem:minClosure}
	Let $F$ be a forest and $\layer_1,\layer_2 \colon V(F) \to \mathbb{N} \cup \{\infty\}$ be two partial strict $H$-decompositions. Let $\layer \colon V(F) \to \mathbb{N} \cup \{\infty\}$ with 
	\begin{align*}
		\layer(v) = \min(\layer_1(v),\layer_2(v))
	\end{align*}
	
	for every $v \in V(F)$. Then, $\layer$ is also a partial strict $H$-decomposition.
\end{lemma}
\begin{proof}
	In order to show that $\layer$ is a partial strict $H$-decomposition, we need to show that it satisfies the properties in \Cref{def:Hdecomp,def:strictH}. Let $V_{<\infty} = \{v \in V(F) \mid \layer(v) < \infty\}$ be the set of nodes that get a finite layer in either $\layer_1$ or $\layer_2$, and let $\vpivot := \{v \in V_{<\infty} \mid \layer(v) \geq \layer(w) \text{ for every $w \in N_F(v)$}\}$ be the set of pivot nodes in $\layer$. Similarly let $\vpivot_1$ be the set of pivot nodes in $\layer_1$.
	
	Consider a node $v \in V_{<\infty}$ and without loss of generality let $\layer(v) = \layer_1(v)$. For each $w\in V$ we have $\layer(w) \le \layer_1(w)$, so using $\layer(v)=\layer_1(v)$ we obtain $\layer(w) \leq \layer_1(w)< \layer_1(v)=\layer(v)$ for all nodes $w \in N_F(v)$ with $\layer_1(w) < \layer_1(v)$. So there can be at most $2$ neighbors $w \in N_F(v)$ with $\layer(w) \ge \layer(v)$ (otherwise $v$ violates this property for $\layer_1$). Hence $\layer$ satisfies \Cref{def:Hdecomp}.   
	
	\begin{claim}
		\label{claim:vlayer1pivot}
		Let $v \in V_{<\infty} \setminus \vpivot$ be a node with $\layer(v) = \layer_1(v)$. Then $v\notin \vpivot_1$ and $v$ satisfies (\ref{cond:pivot}) for $\vpivot_1$ and $\layer_1$.
	\end{claim}
	\begin{proof}
		Since $v\in V_{<\infty}$, we  have that $\layer_1(v)=\layer(v) < \infty$. For contradiction assume that $v \in \vpivot_1$, that is, all nodes $w \in N_F(v)$ satisfy $\layer_1(w) \le \layer_1(v)$. As $\layer(w) \leq \layer_1(w)$ for all nodes $w$, we have that $\layer(w) \le \layer(v)$ for all neighbors $w$ of $v$. This means that $v \in \vpivot$ which contradicts the assumption that $v \in V_{<\infty} \setminus \vpivot$. Therefore, $\layer_1(v) < \infty$ and $v \notin \vpivot_1$, which implies that $v$ satisfies (\ref{cond:pivot}) for $\vpivot_1$ and $\layer_1$, as $\layer_1$ is a valid partial strict $H$-decomposition.
	\end{proof}

	We first define some notation: let $S^{\uparrow} = \{w \in N_F(v) \mid \layer(w) > \layer(v)\}$, $S^{=} = \{w \in N_F(v) \setminus \vpivot \mid \layer(w) = \layer(v)\}$, and $S = S^{\uparrow}\cup S^{=}$. We want to show that (\ref{cond:pivot}) also holds for $\vpivot$ and the function $\layer$. In other words, we want to show that $|S|\leq 1$ holds.

	With the same notation define $S_1^{\uparrow}, S_1^{=}$, and $S_1$ for $\vpivot_1$ and $\layer_1$,  \Cref{claim:vlayer1pivot} states that $|S_1|=|S_1^{\uparrow}\cup S_1^{=}|\leq 1$ holds. 
	We consider the following three cases on how $|S_1|\leq 1$ can be satisfied, and in each of them we show that $|S|\leq 1$ holds. 
	\begin{itemize}
		\item \textbf{Case $1$:} $|S_1| = 0.$ In this case we have $|S_1^{\uparrow}| = 0$ which means $v$ has no neighbors $w \in N_F(v)$ with $\layer_1(w) > \layer_1(v)=\layer(v)$, or in other words $\layer_1(w)\leq \layer(v)$ holds for all $w\in N_F(v)$. As $\layer(w)\leq\layer_1(w)$, this implies that $v$ has no neighbors $w$ with $\layer(w) > \layer(v)$ and hence $|S^{\uparrow}| =0$. 
		
		Moreover, we have $|S_1^{=}| = 0$, which means $v$ has no neighbors $w \in N_F(v) \setminus \vpivot_1$ with $\layer_1(w) = \layer_1(v)$, and as $\layer_1$ is a partial strict $H$-decomposition $v$ can have at most two neighbors $w,w' \in \vpivot_1\cap N_F(v)$ with $\layer_1(w) = \layer_1(w') = \layer_1(v)$. Notice that $w,w'$ are the only two nodes that can potentially belong to $S^{=}$ because for all other nodes $u \in N_F(v) \setminus \{w,w'\}$, we have $\layer(u) \le \layer_1(u) < \layer_1(v) = \layer(v)$.
		The nodes $w,w'$ are only contained in $S^{=}$ if $\layer(w) = \layer_1(w)$, and $\layer(w') = \layer_1(w')$. 
		
		Recall that $w \in \vpivot_1$, so if $\layer(w) = \layer_1(w)$ then $w \in \vpivot$ because all neighbors $u$ of $w$ have $\layer_1(u) \le \layer_1(w)$ which implies $\layer(u) \le \layer(w)$. So we either have $w \in \vpivot$, or $\layer(w) < \layer_1(w) = \layer(v)$, and in both cases $w \notin S^{=}$. The same argument holds for $w'$. Hence $|S^{=}| = 0$. Therefore, we have shown that $|S| = 0$.

		\item \textbf{Case $2$:} $|S_1|=1, |S_1^{=}| = 1.$ In this case we have $|S_1^{\uparrow}| = 0$ which implies that $|S^{\uparrow}| = 0$ by the same argument as Case 1.
		
		Since $|S_1^{=}| = 1$, $v$ has unique neighbor $w \in N_F(v) \setminus \vpivot_1$ with $\layer_1(w) = \layer_1(v)$, and as $\layer_1$ is an $H$-decomposition, $v$ can have at most one neighbor $ w' \in \vpivot_1\cap N_F(v)$ with $\layer_1(w') = \layer_1(v)$. Notice that $w,w'$ are the only two nodes that can potentially belong to $S^{=}$ because for all other nodes $u \in N_F(v) \setminus \{w,w'\}$, we have $\layer(u) \le \layer_1(u) < \layer_1(v) = \layer(v)$. The nodes $w$ and $w'$ are only included in $S^{=}$ if $\layer(w) = \layer_1(w)$ and $\layer(w') = \layer_1(w')$ hold, respectively. 
		
		Recall that $w' \in \vpivot_1$, so if $\layer(w') = \layer_1(w')$ then $w' \in \vpivot$ because all neighbors $u$ of $w'$ have $\layer_1(u) \le \layer_1(w')$ which implies $\layer(u) \le \layer(w')$. So we either have $w' \in \vpivot$, or $\layer(w') < \layer_1(w') = \layer(v)$, and in both cases $w' \notin S^{=}$. Therefore, $v$ has at most one neighbor (i.e. $w$ if $\layer(w) = \layer_1(w)$) in $S^{=}$, and hence, $|S| \le 1$.

		\item \textbf{Case $3$:} $|S_1|=1, |S_1^{\uparrow}| = 1.$ Here we have that $|S_1^{=}| = 0$ and the analysis from Case 1 implies that $|S^{=}| = 0$.

		Therefore we can focus on the unique node $w \in S_1^{\uparrow}$ with $\layer_1(w) > \layer_1(v)$ which is the only node that can potentially belong to $S^{\uparrow}$, as all other neighbors of $u \in N_F(v) \setminus \{w\}$ have $\layer(u) \le \layer_1(u) < \layer_1(v) = \layer(v)$. The node $w$ will contribute at most $1$ to $|S^{\uparrow}|$ no matter what value $\layer(w)$ it assumes, so we have $|S| \le 1$.
	\end{itemize}
	Therefore, $\layer$ also satisfies the additional property of \Cref{def:strictH}, and it is a valid partial strict $H$-decomposition.
\end{proof}

From a high level point of view, \Cref{lem:minClosure} says that we can independently compute two partial strict $H$-decompositions, and even though they might contain conflicting layer assignments for certain nodes, we can obtain a unified decomposition, by assigning each node to the smaller layer of the two choices. In fact, this insight also generalizes to more than two (possibly conflicting) decompositions. At the core of our procedure in \Cref{sec:subtreeRC}, many nodes (independently) learn large parts of the graph. Then, every node computes a partial decomposition on the parts that it has learned, and in a second step all these partial decompositions are combined, where each node takes the minimum layer that it got assigned in any of the decompositions. Taking the minimum is a very efficient procedure in the \mpc model and only requires constant time. The remaining difficulty in \Cref{sec:subtreeRC} is to show that nodes learn large enough parts in the graph in order to make very fast global progress, that is, we show that the unified decomposition assigns a layer to a large fraction of the nodes. 

\section{Strict \texorpdfstring{$H$}{Lg}-decomposition in \mpc}
\label{sec:strict_h}
In this section, we present our $O(\log\log n)$-round \mpc algorithm for computing a strict $H$-decomposition. However, the hardest part of that algorithm, that is, assigning each node that is contained in a small subtree (see definition below) to a layer is deferred to \Cref{sec:subtreeRC}. The algorithm in this section uses $n\cdot \poly\log n$ global space. In \Cref{sec:optimalSpace} we explain how to extend the algorithm to optimal space. 

\subparagraph*{High Level Overview.} For the sake of this high level overview let us first assume that we compute an $H$-decomposition with $O(\log n)$ layers that may not be strict. Similar to the classic rake \& compress algorithm, our algorithm  iteratively assigns nodes to layers. After assigning a node to some layer we \emph{remove} it from the graph and continue on the remaining graph, which may actually become disconnected and turn into a forest. In order to present the details of the high level intuition we require the definition of a subtree which is central to our whole approach.

\begin{definition}[Subtree] \label{def:subtree}
	Let $T$ be a tree. A \emph{subtree} $T'\subseteq T$ is a connected induced subgraph of $T$ such that $T \setminus T'$ contains at most one component. 
	
	A subtree $T'\subseteq F$ of a forest $F$ is a connected induced subgraph of $F$ such that the number of components of $F\setminus T'$ is not larger than the number of connected components of $F$.  
\end{definition}

The definition of a subtree is best understood in a rooted tree, where the subtree rooted at a node $v$ is formed by all its descendants.

In order to assign a layer to all nodes of the graph, we iterate the following two steps until all nodes have received a layer: 

\begin{enumerate}
	\item  Assigns a layer to each node contained in a \emph{small subtree} of size $\leq n^{\delta/10}$ ($\textsc{SubTreeRC}(F)$),
	\item  Assign a layer to each node of degree $\leq 2$ in the remaining graph. 
\end{enumerate}  

This process can be seen as a generalization of the classic rake and compress procedure, in which one iteratively removes leaves, i.e., subtrees of size $1$,  and nodes of degree $2$. The rake and compress procedure requires $O(\log n)$ iterations to \emph{remove} all nodes of the graph. Our generalized process requires $O(1/\delta)=O(1)$ in order to assign a layer to every node of the graph (see \Cref{lem:generalizedRakeCompress} for $x=n^{\delta/10}$). Note that the lemma statement considers a slightly different process than the one presented in this overview; the difference lies in the fact that we actually want to compute a \underline{strict} $H$-decomposition. However, a similar lemma holds for the process of this overview. The main contribution and the main difficulty of our work lies in the procedure $\textsc{SubTreeRC}(F)$ as nodes do not know whether they are contained in a small subtree, but still these subtrees can have diameter up to $n^{\delta/10}$, so conditioned on the 1 vs 2 cycle conjecture it is impossible that a single node can learn the whole subtree in $O(\log\log n)$ rounds (we don't prove this formally, but it's very unlikely that such a result holds without breaking the conjecture).  We explain the details of the procedure $\textsc{SubTreeRC}(F)$ in \Cref{sec:subtreeRC}.

We continue with our generalized rake and compress statement that shows that a constant number of iterations of the aforementioned process suffice.  As we want to compute a strict $H$-decomposition (see \Cref{def:strictH}), we need to slightly modify Step~2 of the above outline, for which we require further definitions; the details of why $\textsc{SubTreeRC}(F)$ returns layers that induce a strict $H$-decomposition are presented in \Cref{sec:subtreeRC}. 

A path in a graph is a \emph{degree-2 path} if all of its nodes, including its endpoints have degree $2$. The \emph{length of a path} is the number of nodes in the path, e.g., a single node is a path of length $1$.  We also need the following claim that states that a tree has more leaves than internal nodes of degree at least three.% (proven in \Cref{app:proofOfWrapper}).

\begin{restatable}{claim}{claimForestInternalNodes}
	\label{claim:forestDegOneDegThree}
	For any forest $F$ we have $|\{v\in V(F)\mid d_F(v)\leq 1\}|\geq |\{v\in V(F)\mid d_F(v)\geq 3\}|$.
\end{restatable}

\begin{proof}
	Let $l$ be the number of nodes with degree at most $1$, $k$ the number of nodes with degree two, and $t$ the number of nodes with degree at least three. 
	We obtain that the graph has at least $(1\cdot l+2\cdot k+3\cdot t)/2$ edges incident to these nodes. At the same time, as $F$ is a forest the total number of edges is upper bounded by $l+k+t-1$. We obtain $l/2+k+3t/2\leq l+k+t-1$, which implies $l\geq t+2$.
\end{proof}

The following lemma is easiest to be understood when setting $x=\ell=1$ where the process (almost) equals the classic rake \& compress process---in fact it consists of a rake step, a compress step, and another rake step---and the theorem shows that it removes $1/3$ of the nodes ($2/3$ of the nodes remain in the graph).
\begin{restatable}[Generalized rake and compress]{lemma}{lemprogress}
	\label{lem:generalizedRakeCompress}
	Let $x, \ell \in \mathbb{Z}$. Consider a process on a tree $T$ that consists of the following steps: 
	
	\begin{enumerate} 
		\item Remove (at least) all subtrees of size $\leq x$ from $T$, resulting in $T_1$, 
		\item Remove (at least) all nodes contained in a degree-2 path of length at least $\ell$ from $T_1$, resulting in $T_2$,
		\item Remove (at least) all nodes with degree $\leq 1$ from $T_2$.
	\end{enumerate}
	
	The number of nodes remaining is at most a $1/(1+(x+1)/2\ell) = O(\ell/x)$ fraction of the nodes from $T$.
	The degrees of nodes in Step 2) and 3) of the process are with respect to the graph induced by remaining nodes at the respective step.
\end{restatable}
\begin{proof}
	Let $x,l \in \mathbb{Z}$ and $T$ be a tree. Let $\mathcal{T}=\{t\subseteq T \mid t\text{ is subtree of $T$}, |V(t)|\leq x\}$ be the set of all subtrees of size $\leq x$ of $T$. Let $T_1$ be the graph after the removal of all subtrees of size $\leq x$, that is, $T_1$ is the graph induced by $V\setminus \big(\bigcup_{t\in\mathcal{T}} V(t)\big)$. Let $T_2$ be the graph obtained after Step 3 of the process. Denote $n=|V(T)|$, $n_1=|V(T_1)|$ and $n_2=|V(T_2)|$.
	
	Let $B=\bigcup_{t\in\mathcal{T}} V(t)$ be the set of nodes that are part of a subtree of size $\leq x$. We refer to these nodes as blue nodes. We refer to every leaf node in $T_1$ as a red node. 
	
	Observe that for each red node $v$ there is a non-empty collection of subtrees $\emptyset\neq \mathcal{T}_v\subseteq \mathcal{T}$ such that $v$ is adjacent to the unique root of each $t\in \mathcal{T}_v$ in the original tree $T$. Note that these subtrees contain blue nodes only and were removed in the first step of the algorithm. We obtain that for each red node $v$ the union of the subtrees $\mathcal{T}_v$ contains at least $x$ nodes, as otherwise $v$ together with the subtrees in $\mathcal{T}_v$ would form a subtree of size at most $x$, implying that $v$  would not remain in $T_1$ as a red node. Additionally, for distinct red nodes $v$ and $v'$ the sets $\bigcup_{t\in \mathcal{T}_v}V(t)$ and $\bigcup_{t\in \mathcal{T}_{v'}}V(t)$ are disjoint. 
	
	As Step 3 strips away all these red nodes (they have degree $1$ in $T_1$ and are for sure removed in the third step of the process), for each red node in $T_1$ the complete process removes $x+1$ nodes from $T$ (we aren't yet using that Step 2 removes certain nodes of $2$ from $T_1$) and we obtain the following. 
	\begin{align*}
		n &>  n_2 +  (x+1) \cdot |\{v\in T_1: \deg_{T_1}(v)=1\}|  \\
		& \stackrel{*}{\geq} n_2 +  (x+1) \cdot |\{v\in T_1: \deg_{T_1}(v)\neq 2\}|/2\stackrel{**}{\geq}n_2\cdot (1+(x+1)/2\ell).
	\end{align*}
	At $*$ we used \Cref{claim:forestDegOneDegThree}, that is, that  any forest contains fewer degree $\geq 3$ nodes than nodes of degree $1$. At $**$ we finally use that in Step~2 we remove certain nodes with degree $2$ from $T_1$. The reasoning is slightly more involved. In fact, after Step~2 the number of nodes is upper bounded by $\ell\cdot \{v\in T_1| \deg_{T_1}(v)\neq 2\}$, as only nodes with degree $\geq 3$ and 1, and paths of length at most $\ell-1$ between them remain in the graph; as Step~3 can only remove additional nodes we obtain $n_2\leq \ell \cdot |\{v\in T_1| \deg_{T_1}(v)\neq 2\}|$, which transforms to $|\{v\in T_1| \deg_{T_1}(v)\neq 2\}|\geq n_2/\ell$, which we use at $**$.
	Dividing both sides of the inequality by $(1+(x+1)/2\ell) > 1$ we deduce that $n_2\leq n/(1+(x+1)/2\ell)$.
\end{proof}

We now state a lemma for a key subroutine that we will use as black box in this section and dedicate \Cref{sec:subtreeRC} to designing an algorithm that proves the lemma.

\begin{restatable}[\textsc{SubTreeRC}]{lemma}{lemsubtreerc}
	\label{lem:mpc_subtree_removal}
	Let $F$ be a forest on $n$ vertices. There exists a deterministic \mpc algorithm \textsc{SubTreeRC} with $O(n^{\delta})$ local space,  $0 < \delta < 1$, and $\widetilde{O}(n)$ global space which takes $F$ as input and computes in $O(\log \log n)$ rounds a partial strict $H$-decomposition $\layer \colon V(F) \mapsto [\lceil \log(|V(F)| + 1)\rceil] \cup \{\infty\}$ such that $\layer(v) < \infty$ for every node $v \in V(F)$ contained in a subtree of size $n^{\delta/10}$.
\end{restatable}

Our \mpc algorithm for computing strict $H$-decomposition appears in \Cref{alg:strictHDecom}. We will now prove the correctness and progress guarantees of our algorithm.

\begin{algorithm}[t]
	\caption{Strict $H$-decomposition.}\label{alg:strictHDecom}
	\begin{algorithmic}[1]
		
		\State{Throughout  $V_{\infty}=\{v\in F \mid \layer(v)=\infty\}$ denotes the set of nodes whose layer equals $\infty$.} 
		\Function{StrictHDecomp}{Forest $F$}
		\State{\textbf{Initialize:} $\layer(v)=\infty$ for all $v\in V(F)$; $\mathsf{offset}\larr  \lceil \log (|V(F)| + 1)\rceil+1$}
		\For{$i = 1,2,\ldots,\lceil 10/\delta\rceil$}
		\State{$F_{i}\larr F[V_{\infty}]$}
		\State{$\layer\larr i\cdot \mathsf{offset}+\textsc{SubTreeRC}(F_i, x = n^{\delta/10}$)}  \phantomsection\label{line:removeSubtrees}
		\State{Let $\vpivot_i\larr \{v\in V_{\infty} \mid d_{V_{\infty}}(v)\leq 2, d_{V_{\infty}}(w)\leq 2 \text{ for all } w\in N(v)\}$}
		\State{$\layer(v)\larr (i+1)\cdot \mathsf{offset}$ for every node $v\in \vpivot_i$ \phantomsection\label{line:removePivotmain}}
		\State{$\layer(v)\larr (i+1)\cdot \mathsf{offset}$ for every node $v\in V_{\infty}$ with $\leq 1$ in $V_{\infty}$ \phantomsection\label{line:removeDegOnemain}}
		\EndFor
		\State{\Return{$\layer$}}
		\EndFunction
	\end{algorithmic}
\end{algorithm}

%The proof of both of these lemmas appear in \Cref{app:proofOfWrapper}
\begin{restatable}{lemma}{lemstrictHCorrect}\label{lem:strictH-correctness}
	At the end of each iteration $i$, we have that $\layer$ is a partial strict $H$-decomposition with at most $(i+1)\cdot \mathsf{offset}$ layers.
\end{restatable}

\begin{proof}
	We prove the lemma inductively over the iterations $i$. So assume that at the beginning of iteration $i$, $\layer$ is a correct partial strict $H$-decompositon with at most $i \cdot \mathsf{offset}$ layers. In \Cref{line:removeSubtrees,line:removePivotmain,line:removeDegOnemain} we will assign a layer greater than $i \cdot \mathsf{offset}$ to some nodes in $v\in V_{\infty}$ with $\layer(v) = \infty$.
	Given a partial strict $H$-decomposition, if we compute another partial strict $H$-decomposition of some of the nodes in $V_{\infty}$ such that they obtain a finite layer that is larger than all nodes in $V \setminus V_{\infty}$, we still have a valid partial strict $H$-decomposition. Therefore all we need to show is that the layers assigned to nodes in iteration $i$ form a partial strict $H$-decomposition.
	
	By \Cref{lem:mpc_subtree_removal}, \Cref{line:removeSubtrees} computes a partial strict $H$-decomposition with largest layer $< (i+1)\cdot \mathsf{offset}$, and in \Cref{line:removePivotmain,line:removeDegOnemain} we assign layer $(i+1)\cdot \mathsf{offset}$ to nodes that are in $V_{\infty}$ after \Cref{line:removeSubtrees}. Therefore, by the previous argument, we now need to prove that nodes assigned a layer in \Cref{line:removePivotmain,line:removeDegOnemain} form a partial strict $H$-decomposition.
	
	Nodes $v \in \vpivot_i$ that are assigned a layer in \Cref{line:removePivotmain} have degree $2$ by definition. Both the neighbors of $v$ are assigned layer either in \Cref{line:removePivotmain} or \Cref{line:removeDegOnemain}, as they are either in $\vpivot_i$ or their degree becomes $1$ after $v$ is assigned a layer. Therefore $\vpivot_i \subseteq \vpivot$ of \Cref{def:strictH}.
	
	Finally, nodes $v \notin \vpivot_i$ that are assigned a layer in \Cref{line:removeDegOnemain} have at most one neighbor in $N_F(v) \setminus \vpivot$ with the same or higher layer, since they have degree at most $1$ after removing nodes in $\vpivot_i$. This implies that the partial strict $H$-decomposition properties are satisfied by all nodes. Since the largest finite layer assigned in iterations $1, \dots, i$ is $(i+1)\cdot \mathsf{offset}$, this proves the statement of the lemma.
\end{proof}

\begin{restatable}{lemma}{lemPivotIsRake}\label{lem:strictH-generalizedrc-connect}
	In iteration $i$, \Cref{line:removeSubtrees,line:removePivotmain,line:removeDegOnemain} correspond to a generalized rake and compress step with $x = n^{\delta/10}$ and $\ell = 3$.
\end{restatable}

\begin{proof}
	\Cref{lem:mpc_subtree_removal} gives us that \Cref{line:removeSubtrees} removes all subtrees of size at most $x = n^{\delta/10}$. Furthermore, \Cref{line:removePivotmain,line:removeDegOnemain} together remove all nodes contained in a degree-$2$ path of length at least $3$. To see this, let $P$ be a maximal path of degree-2 nodes with length $\ell \ge 3$. All nodes in $P$ except possibly the two end points belong to $\vpivot_i$, and they will receive layer $(i+1)\cdot \mathsf{offset}$ in \Cref{line:removePivotmain}. Note that the two end points of a path $P$ must have degree at most $2$ with one neighbor in $P$, and possibly another neighbor not in $P$ having degree at least $3$ (as otherwise $P$ is not maximal). The two end points become degree $1$ nodes after removing the nodes layered in \Cref{line:removePivotmain}, and hence they will receive layer $(i+1)\cdot \mathsf{offset}$ in \Cref{line:removeDegOnemain}.
	
	Removing the end points of $P$ does not create additional degree $1$ nodes, so \Cref{line:removeDegOnemain} also removes all with degree $\leq 1$ after we remove the subtrees of size at most $x$ and paths of length at least $\ell$.
\end{proof}

\begin{corollary} \label{cor:strictH-single-step-progress}
	In iteration $i$, \Cref{line:removePivotmain,line:removeDegOnemain} correspond to a generalized rake and compress step with $x = 1$ and $\ell = 3$.
\end{corollary}

\begin{theorem}
	\label{thm:strictHDecomppolylogMemory}
	Algorithm $\textsc{StrictHDecomp}(F)$ (\Cref{alg:strictHDecom}) applied to some forest $F$ computes a strict $H$-decomposition of $F$ with $O(\log n)$ layers, uses $O(\log\log n)$ low-space \mpc rounds and $\widetilde{O}(n)$ global space.
\end{theorem}
\begin{proof}
	By \Cref{lem:strictH-generalizedrc-connect,lem:generalizedRakeCompress}, in each iteration, the number of nodes in the forest shrinks by a factor of $O(n^{\delta/10})$. Therefore, after $O(1/\delta)$ iterations of the for loop, the number of nodes with layer $\infty$ will be zero.
	
	By \Cref{lem:strictH-correctness}, in an iteration $i$, we produce a partial strict $H$-decomposition with at most $(i+1) \cdot \mathsf{offset}$ layers and in the next iterations $j > i$, we compute a partial strict $H$-decomposition of the nodes that received layer $\infty$ ($V_{\infty}$) in iteration $i$. The $\mathsf{offset}$ value ensures that the nodes in $V_{\infty}$ get a higher layer than the nodes in $V \setminus V_{\infty}$. After $i = O(1/\delta)$ iterations, each node has a layer at most $O(\log n)$ since $(i+1) \cdot \mathsf{offset} = O(\log n)$, and hence we produce a valid strict $H$-decomposition.
	
	\Cref{lem:mpc_subtree_removal} ensures that implementing each iteration takes $O(\log\log n)$ low-space \mpc rounds and $\widetilde{O}(n)$ global space. The theorem follows because there are just $O(1/\delta)$ iterations.
\end{proof}

\section{Massively Parallel Subtree Rake and Compress}
\label{sec:subtreeRC}

This section is dedicated to designing an algorithm that proves \Cref{lem:mpc_subtree_removal}, which states that we can compute in $O(\log \log n)$ rounds a partial strict $H$-decomposition that assigns each node contained in a subtree of size $O(n^{\delta/10})$ to one of $O(\log n)$ layers. We restate the lemma.

\lemsubtreerc*

Our algorithm critically relies on the balanced graph exponentiation technique mentioned in \Cref{sec:contributions} and explained in detail in \cref{sec:balanced_exponentiation}. In the following definition, you should think about $U$ as being the set of nodes that $v$ has stored in its local memory after the balanced graph exponentiation. We refer to $U$ as good if it contains all nodes within distance $O(\log n)$ of $v$, except for potentially one direction, for which no node is contained in $U$.

\begin{definition}[$U$ is a good subset for $v$]
	\label{def:good_subset}
	Let $F$ be a forest, $U \subseteq V(F)$ and $v \in V(F)$. We say that $U$ is a good subset for $v$ if 
	
	\begin{enumerate}
		\item $v \in U$,
		\item $|N_F(v) \setminus U| \leq 1$, i.e., $v$ has at most one neighbor in $F$ not in $U$,
		\item $N_F(w) \subseteq U$ for every $w \in U \setminus \{v\}$ with $d_F(v,w) \leq 3L$ where $L := \lceil\log(|U| + 1) \rceil$.
	\end{enumerate}
	
\end{definition}

In \cref{sec:conservative_peeling}, we give a peeling algorithm that takes as input a set $U$ and computes a partial strict $H$-decomposition with $O(\log n)$ layers by repeatedly peeling off low-degree vertices contained in $U$. Moreover, if $U$ is good for $v$, then $v$ gets assigned to one of the $O(\log n)$ layers. This peeling algorithm will later be simulated without any further communication on the machine that has stored the set $U$ in its memory.

Using the balanced graph exponentiation technique, we can compute a collection of sets $U_1,U_2,\ldots,U_k$ such that for each node $v$ contained in a subtree of size at most $n^{\delta/10}$ there exists some subset $U_j$ that is good for $v$. In particular, in \cref{sec:proofOfLemma} we prove the following statement.

\begin{lemma}[Lemma from Balanced Exponentiation]
	\label{lem:subsets}
	Let $F$ be a forest on $n$ vertices. There exists a deterministic low-space \mpc algorithm with $O(n^\delta)$ local space, $0 < \delta < 1$, and $O(n \cdot \poly(\log n))$ global space which takes $F$ as input and computes in $O(\log \log n)$ rounds a collection of non-empty sets $U_1,U_2,\ldots,U_k \subseteq V(F)$ such that 
	
	\begin{enumerate}
		\item (Local Space) $|U_j| = O(n^\delta)$ for every $j \in [k]$, 
		\item (Global Space) $\sum_{j=1}^k |U_j| = O(n \cdot \poly(\log n))$ and 
		\item for every $v \in V(F)$ which is contained in a subtree of size at most $n^{\delta/10}$, there exists a $j \in [k]$ such that $U_j$ is a good subset for $v$ (see \cref{def:good_subset}).
	\end{enumerate}
	
	Moreover, the algorithm also computes for each $U_j$ the forest $F[U_j]$ induced by vertices in $U_j$ and stores it on a single machine.
\end{lemma}

Our final \mpc algorithm for proving \cref{lem:mpc_subtree_removal} first computes a collection of sets $U_1,U_2,\ldots,U_k$ using \cref{lem:subsets}. Then, the machine storing $U_j$ locally simulates the peeling algorithm of \cref{sec:conservative_peeling} with input $U_j$. As a result, we obtain one partial strict $H$-decomposition for each set $U_j$. These partial strict $H$-decompositions are then combined into one partial strict $H$-decomposition by assigning each node to the smallest layer assigned by any of the partial strict $H$-decompositions. More details can be found in \cref{sec:subtree_rc}.

\subsection{The Conservative Peeling Algorithm}

\label{sec:conservative_peeling}

\cref{alg:subsetrc} computes a partial strict $H$-decomposition by repeatedly removing low-degree vertices contained in $U$.

\begin{algorithm}[H]
	\caption{Conservative Peeling Algorithm.}\label{alg:subsetrc}
	\begin{algorithmic}[1]
		
		\Function{ConservativePeeling}{Forest $F$, Subset $U \subseteq V(F)$}
		
		\State{$V_{\ge 1} \larr V(F)$, $\layer \colon V(F) \mapsto \mathbb{N} \cup \{\infty\}$, $L \larr \lceil\log(|U| + 1)\rceil$} 
		\For{$i = 1,2,\ldots,L$}
		\State{We define $N_{\geq i}(v) := N_{F[V_{\geq i}]}(v)$ for every $v \in V_{\geq i}$.}
		\State{$\vpivot_i \larr \{v \in V_{\geq i} \cap U \mid N_{\geq i}(v) \subseteq U \text{ and $\forall~w \in N_{\geq i}(v) \cup \{v\}$: }|N_{\geq i}(w)|\leq 2\}$.}\phantomsection\label{line:subsetrc-remove-pivots}
		\State{$V_i \larr \vpivot_i \cup \{v \in V_{\geq i} \cap U \mid |N_{\geq i}(v) \setminus \vpivot_i| \leq 1\}$}\phantomsection\label{line:subsetrc-vi}
		\State{$V_{\ge i+1} \larr V_{\ge i} \setminus V_i$}
		\State{$\layer(v) \larr i$ for every $v \in V_i$}
		\EndFor
		\State{$\layer(v) \larr \infty$ for every $v \in V_{\geq L + 1}$}
		\State{\Return{$\layer$}}
		
		\EndFunction
		
	\end{algorithmic}
\end{algorithm}

If we would just be interested in computing a partial $H$-decomposition instead of a strict one, then we could replace \cref{line:subsetrc-remove-pivots,line:subsetrc-vi} with the single line $V_i \larr \{v \in V_{\geq i} \cap U \mid |N_{\geq i}(v)| \leq 2\}$.

We first show that \cref{alg:subsetrc} indeed computes a partial strict $H$-decomposition.

\begin{lemma}
	\label{lem:conservative_returns_strict}
	Let $F$ be a forest and $U \subseteq V(F)$. Let $\layer \colon V(F) \mapsto \mathbb{N} \cup \{\infty\}$ be the mapping computed by \cref{alg:subsetrc} when given $F$ and $U$ as input. Then, $\layer$ is a partial strict $H$-decomposition as defined in \cref{def:strictH}.
\end{lemma}
\begin{proof}
	As in \cref{def:strictH}, we define $V_{<\infty} := \{v \in V(F) \mid \layer(v) < \infty\}$ and 
	\begin{align*}
		\vpivot := \{v \in V_{<\infty} \mid \layer(v) \geq \layer(w) \text{ for every $w \in N_F(v)$}\}.
	\end{align*}
	
	We first show that $\bigcup_i \vpivot_i \subseteq \vpivot$. To that end, consider an arbitrary $v \in \vpivot_i$. We have $\layer(v) = i$. Thus, we have to show that for a given $w \in N_F(v)$, it holds that $\layer(w) \leq i$. We only have to consider the case that $w \in V_{\geq i}$; otherwise $\layer(w) < i$. Thus, $w \in N_{\geq i}(v)$ which together with $v \in \vpivot_i$ implies $|N_{\geq i}(w)| \leq 2$ and $w \in U$. Moreover, $v \in N_{\geq i}(w) \cap \vpivot_i$ and therefore  
	\begin{align*}
		|N_{\geq i}(w) \setminus \vpivot_i| = |N_{\geq i}(w)| - |N_{\geq i}(w) \cap \vpivot_i| \leq 2 - 1 = 1, 
	\end{align*}
	
	which together with $w \in V_{\geq i} \cap U$ directly gives $\layer(w) = i$. Thus, we indeed have shown that $\bigcup_i \vpivot_i \subseteq \vpivot$.
	Next, consider an arbitrary $v \in V_{< \infty}$. We have to show that 
	\begin{align*}
		|\{w \in N_F(v) \mid \layer(w) \geq \layer(v)\}| \leq 2.
	\end{align*}
	
	Put differently, we have to show that $|N_{\geq i}(v)| \leq 2$ where $i := \layer(v)$. Note that if $v \in \vpivot_i$ or $v$ neighbors a node in $\vpivot_i$, then we directly get $|N_{\geq i}(v)| \leq 2$ from the way $\vpivot_i$ is defined. On the other hand, if $v \notin \vpivot_i$ and $|N_{\geq i}(v) \cap \vpivot_i| = \emptyset$, then we even get the stronger property $|N_{\geq i}(v)| \leq 1$. 

 Thus, it remains to show that for a given $v \in V_{<\infty} \setminus \vpivot$, it holds that 

 \[|\{w \in N_F(v) \setminus \vpivot \mid \layer(w) = \layer(v)\} \cup \{w \in N_F(v) \mid \layer(w) > \layer(v)\}| \leq 1.\]

	Let $i := \layer(v)$. As $v \notin \vpivot$ and $\vpivot_i \subseteq \vpivot$, we get $v \in V_i \setminus \vpivot_i$ and therefore $|N_{\geq i}(v) \setminus \vpivot_i| \leq 1$, as $v \in V_i \setminus \vpivot_i$. We have

 \begin{align*}
 &|\{w \in N_F(v) \setminus \vpivot \mid \layer(w) = \layer(v)\} \cup \{w \in N_F(v) \mid \layer(w) > \layer(v)\}| \\
 &\leq  |\{w \in N_F(v) \setminus \vpivot_i \mid \layer(w) = \layer(v)\} \cup \{w \in N_F(v) \mid \layer(w) > \layer(v)\}| \\
 % &\leq  |N_{\geq i}(v) \setminus \vpivot_i| \leq 1. \qedhere
 &\leq  |N_{\geq i}(v)| \leq 1. \qedhere
 \end{align*}

\end{proof}

Next, we show that if $U$ is a good subset for $v$, as defined in \cref{def:good_subset}, then $v$ gets assigned to one of the $O(\log n)$ layers.

\begin{lemma}
	\label{lem:conservative_progress}
	Let $F$ be a forest, $U \subseteq V(F)$ and $v \in V(F)$. Let $\layer \colon V(F) \mapsto \mathbb{N} \cup \{\infty\}$ be the mapping computed by \cref{alg:subsetrc} when given $F$and $U$ as input. If $U$ is a good subset for $v$ (see \cref{def:good_subset}), then $\layer(v) \leq L := \lceil \log(|U| + 1)\rceil$.
\end{lemma}
\begin{proof}
	In the following, we define $T$ as the connected component which contains $v$ in the graph $F[U]$. Throughout the proof, we think of $T$ as being rooted towards $v$. For a node $u \in V(T)$, we denote by $T_u$ the subtree of $T$ rooted at $u$ and by $V^{children}(u)$ the set of children of $u$. The proof splits into two parts. The first part shows the following for a given node $w \in V(T)$ with $d_T(v,w) \leq 3L - 2$ and some $i \in [L]$: If $w$ has at most one child, at most one grandchild and at most one great-grandchild in $V_{\geq i}$, then $\layer(w) \leq i$. The second part then uses an inductive argument on top of the first part to prove that if a node $w \in V(T)$ with $\layer(w) \geq i$ is sufficiently close to $v$, then the subtree rooted at $w$ has $\Omega(2^i)$ nodes. As the subtree rooted at $v$ can trivially have at most $|U|$ nodes, this allows us to conclude that $\layer(v) \leq L$. 
	
	To prove the first part, first note that for every node $u \in V(T)$ with $d_T(v,u)\leq 3L$, it holds that $|N_F(u) \setminus V^{children}(u)|\leq 1$. Depending on whether $u = v$ or $u \neq v$, this follows from either the second or the third property of \cref{def:good_subset}. In particular, if $u \in V_{\geq i}$ and $u$ has at most one child in $V_{\geq i}$, then $|N_{\geq i}(u)| \leq 2$. 
	
	Now, consider some $w \in V(T)$ with $d_T(v,w) \leq 3L - 2$ and some $i \in [L]$. Assume that $w$ has at most one child, at most one grandchild and at most one great-grandchild in $V_{\geq i}$. Let $u \in V_{\geq i}$ be either $w$ itself, a child or a grandchild of $w$. By the triangle inequality, we have $d_T(v,u) \leq d_T(v,w) + d_T(w,u) \leq (3L-2) +2 = 3L$. Furthermore, by our assumption $u$ has at most one child in $V_{\geq i}$, and therefore we can conclude $|N_{\geq i}(u)| \leq 2$. In particular, for a given $c \in V^{children}(w) \cap V_{\geq i}$, it holds that $|N_{\geq i}(u)| \leq 2$ for every $u \in N_{\geq i}(c) \cup \{c\}$. Therefore, $c \in \vpivot_i$. In particular, $V^{children}(w) \cap V_{\geq i} \subseteq \vpivot_i$. Thus, if $w \in V_{\geq i}$, we get
	\[|N_{\geq i}(w) \setminus \vpivot_i| \leq |N_{\geq i}(w) \setminus (V^{children}(w) \cap V_{\geq i})| = |N_{F}(w) \setminus V^{children}(w)| \leq 1,\]
	which allows us to conclude that $\layer(w) \leq i$. This finishes the proof of the first part of the argument. 
	
	Next, we prove by induction that for every $i \in \{0,1,\ldots,L\}$ and every $w \in V(T)$ with $d_T(v,w) \leq L - 3i$ and $\layer(w) > i$, it holds that $|V(T_w)| \geq 2^i$.
	The base case $i = 0$ trivially holds as $|V(T_w)| \geq 2^0$ for every $w \in V(T)$. For the induction step, consider some fixed $i \in [L]$ and some $w \in V(T)$ with $d_T(v,w) \leq L - 3i$ and $\layer(w) > i$. Using the first part, this implies that $w$ has two children, or two grandchildren, or two great-grandchildren in $V_{\geq i}$. In particular, there are two nodes $u_1,u_2 \in V_{\geq i} \cap V(T_w)$ with $d_T(w,u_j) \leq 3$ for $j \in [2]$ and $V(T_{u_1}) \cap V(T_{u_2}) = \emptyset$.  In particular, $d_T(v,u_j) \leq d_T(v,w) + d_T(w,u_j) \leq L - 3(i-1)$ and $\layer(u_j) > i-1$ and thus we get by induction that $|V(T_{u_j})| \geq 2^{i-1}$ for $j \in [2]$. As $V(T_{u_1}) \cap V(T_{u_2}) = \emptyset$, we therefore get $|V(T_w)| \geq 2^i$, which finishes the induction.
	In particular, as $|V(T_v)| < 2^L$, we can conclude that $\layer(v) \leq L$.
\end{proof}

Finally, we show that we can locally simulate \cref{alg:subsetrc} by only knowing the forest induced by vertices in $U$ and the degree of each node $U$ in the original forest.

\begin{lemma}[Local Sequential Simulation]
	\label{lem:local_simulation}
	Let $F$ be an arbitrary forest and $U \subseteq V(F)$ be a non-empty subset. Let $\layer \colon V(F) \mapsto \mathbb{N} \cup \{\infty\}$ be the mapping computed by \cref{alg:subsetrc} when given $F$ and $U$ as input. There exists a sequential algorithm running in $O(|U|)$ space with the following guarantee: The input of the algorithm is the forest $F[U]$ and the degree $\deg_F(u)$ of each node $u \in U$ in the forest $F$. The algorithm outputs for each node $v \in U$ its layer $\layer(v)$.
\end{lemma}
\begin{proof}
	Note that it suffices to show the following: Fix some $i \in [L]$ and assume we know for each node $v \in U$ whether $v$ is contained in $V_{\geq i}$. Then, we can compute in $O(|U|)$ space for each node $v \in U$ whether $v \in V_i$. Consider some arbitrary node $v \in V_{\geq i} \cap U$. First, note that can determine whether $N_{\geq i}(v) \subseteq U$ by simply checking whether $\deg_F[v] = \deg_{F[U]}(v)$. If $\deg_F[v] = \deg_{F[U]}(v)$, then $N_{\geq i}(v) \subseteq N_{F}(v) \subseteq U$. On the other hand, if $\deg_{F[U]}(v) < \deg_F(v)$, then $v$ has a neighbor $w \in N_F(v)$ which is not contained in $U$. As $V_{i'} \subseteq U$ for every $i' \in [L]$, it follows that $w \in V_{\geq i}$ and therefore $N_{\geq i}(v) \not \subseteq U$. By a similar reasoning, we get that
	\begin{align*}
		|N_{\geq i}(v)| = |N_{\geq i}(v) \cap U| + \deg_{F}(v) - \deg_{F[U]}(v).
	\end{align*}
	
	Therefore, $v$ can compute $|N_{\geq i}(v)|$ by just counting how many of its neighbors in $U$ are contained in $V_{\geq i}$. Thus, we can compute in $O(|U|)$ space for each node $v \in U$ whether it is contained in $\vpivot_i$. Afterwards, we can compute, again in $O(|U|)$ space, for each node $v \in U$ whether it is contained in $V_i$.
\end{proof}

\subsection{Subtree Rake and Compress}
\label{sec:subtree_rc}

\cref{alg:rc} computes a partial strict $H$-decomposition with $O(\log n)$ layers where each node in a subtree of size at most $x$ is assigned to one of the layers. We later set $x= n^{\delta/10}$.
The correctness follows from the key structural property that partial strict $H$-decompositions are closed under taking minimums (\cref{lem:minClosure}). 
\begin{algorithm}[H]
	\caption{\textsc{SubTreeRC} Algorithm.}\label{alg:rc}
	\begin{algorithmic}[1]
		
		\Function{\textsc{SubTreeRC}}{forest $F$, $x \in \mathbb{N}$}
		\State{Let $U_1,U_2,U_3,\ldots,U_k \subseteq V(F)$ such that for every node $v \in V(F)$ contained in a subtree of size at most $x$ in $F$, there exists some $j \in [k]$ such that $U_j$ is a good subset for $v$ (see \cref{def:good_subset})}
		\State{$\layer_j \larr ConservativePeeling(F,U_j)$ for every $j \in [k]$ \Comment{$\layer_j \colon V(F) \mapsto \mathbb{N} \cup \{\infty\}$}}	
		\State{$\layer(v) = \min_{j \in [k]}\layer_j(v)$}\Comment{$\layer \colon V(F) \mapsto \mathbb{N} \cup \{\infty\}$} 
		\State{\Return{$\layer$}}	
		\EndFunction	
	\end{algorithmic}
\end{algorithm}
\begin{lemma}
	\label{lem:subtree_correctness}
	The algorithm above computes a partial $H$ decomposition $\layer \colon V(F) \mapsto [\lceil \log(|V(F)| + 1)\rceil] \cup \{\infty\}$ such that $\layer(v) < \infty$ for every node $v \in V(F)$ contained in a subtree of size at most $x$.
\end{lemma}
\begin{proof}
	\Cref{lem:conservative_returns_strict} states that $\layer_j$ is a strict partial $H$-decomposition for every $j \in [k]$. Hence, \Cref{lem:minClosure} implies that $\layer$ is also a strict $H$-decomposition. Moreover, for every node $v \in V(F)$ contained in a subtree of size at most $x$ in $F$, there exists some $j \in [k]$ such that $U_j$ is a good subset for $v$. Thus, \Cref{lem:conservative_progress} gives that $\layer_j(v) < \infty$ and therefore $\layer(v) < \infty$.
\end{proof}

We are now ready to prove \cref{lem:mpc_subtree_removal}.

\begin{proof}[Proof of \cref{lem:mpc_subtree_removal}]
	We first run the balanced exponentiation algorithm of \cref{lem:subsets} which runs in $O(\log \log n)$ rounds and needs $\widetilde{O}(n)$ global space. As a result, we obtain a collection of non-empty subsets $U_1,U_2,\ldots,U_k \subseteq V(F)$ satisfying the three properties stated in \cref{lem:subsets}. In particular, for each $j \in [k]$, there exists one machine which has stored $F[U_j]$. As $|U_j| = O(n^\delta)$ and $F$ is a forest, $F[U_j]$ indeed fits into one machine. Moreover, one can compute in $O(1)$ rounds for each node $v \in V(F)$ its degree $\deg_F(v)$ and store $\deg_F(v)$ for every node $v \in U_j$ in the same machine as we store $F[U_j]$ using standard \mpc primitives \cite{goodrich}. Let $\layer_j \larr ConservativePeeling(F,U_j)$. \cref{lem:local_simulation} implies that we can compute $\layer_j(u)$ for every node $u \in U_j$ locally on the machine that stores $F[U_j]$ without any further communication. Then, in $O(1)$ rounds we can compute $\layer(v) =  \min_{j \in [k]} \layer_j(v)$ for every $v \in V$ using the fact that we can sort $N$ items in $O(1)$ rounds in the low-space \mpc model with $\widetilde{O}(N)$ global space \cite{goodrich}. In more detail, we create one tuple $(v,\layer_j(v))$ for every $j \in [k]$ and $u \in U_j$ and one tuple $(v,\infty)$ for every node $v \in V(F)$. Then, we sort the tuples according to the lexicographic order. Given the sorted tuples, it is straightforward to determine $\layer(v)$ for every $v \in V$. As $\sum_{j = 1}^k |U_j| = \widetilde{O}(n)$, it follows that the algorithm needs $\widetilde{O}(n)$ global space. It thus remains to argue about the correctness, which directly follows from the third property of \cref{lem:subsets} and \cref{lem:subtree_correctness}.
\end{proof}

\section{Coloring, MIS, and Matching}
\label{sec:coloring}

The following theorem is proven at the end of the section.
\begin{theorem} \label{thm:bound3coloring}
	There is a deterministic $O(\log \log n)$ round algorithm for $3$-coloring trees in the low-space \mpc model using $\widetilde{O}(n)$ words of global space.
\end{theorem}

For an input tree $F$, consider having a strict $H$-decomposition $\layer: V(F) \to \mathbb{N}$ described in \Cref{def:strictH}, which we get from \Cref{alg:strictHDecom} in $O(\log \log n)$ rounds and $\widetilde{O}(n)$ words of global space. We first color the subgraph induced by the nodes in $\vpivot$. Recall that $\vpivot$ is the set of nodes that have no neighbor with a higher layer.

\subparagraph*{Coloring the Pivot Nodes.}
The subgraph $F[\vpivot]$ has maximum degree $2$ each node $v \in \vpivot$ has at most two neighbors in $\vpivot$ with the same layer, and no neighbors with higher layer. In order to color $F[\vpivot]$, we first run Linial's $O(\Delta^2)$-coloring algorithm \cite{linial}, which requires $O(\log^* n)$ rounds. Since $\Delta(F[\vpivot]) \le 2$, Linial's algorithm results in an $O(1)$-coloring which we can convert to a $3$-coloring by performing the following: In each round, all nodes with the highest color among their neighbors in $\vpivot$ recolor themselves with the smallest color such that a proper coloring is preserved. Clearly, one color is eliminated in each round and since each node $v$ has at most $2$ neighbors in $F[\vpivot]$, we achieve a $3$-coloring of $F[\vpivot]$ in a constant number of rounds.

\subparagraph*{Coloring the Remaining Nodes.}
We will now compute a $3$-coloring of the nodes in $V \setminus \vpivot$. We first orient all edges $e = \{u,v\}$ with $u, v \in V \setminus \vpivot$ from $u$ to $v$ if $\layer(u) < \layer(v)$ and arbitrarily if $\layer(u) = \layer(v)$. The following lemma will help us to ensure that we do not create conflicts with the $3$-coloring computed on $\vpivot$.

\begin{lemma}
	Each node in $v \in V \setminus \vpivot$ has at most two forbidden colors. If $v$ has an outgoing edge, then it can have at most one forbidden color.
\end{lemma} 
\begin{proof}
	Each node in $v \in V \setminus \vpivot$ has at most two neighbors in $\vpivot$. This is because nodes in $\vpivot$ do not have neighbors in higher layer, so $v$ can only have neighbors in $\vpivot$ at the same or higher layer. By \Cref{def:strictH}, $v$ can have at most two such neighbors.
	
	Nodes $v$ with one outgoing edge can have at most one neighbor in $\vpivot$, as otherwise $v$ has three neighbors with same or higher layer, and \Cref{def:strictH} is violated. So if $v$ has an outgoing edge, it can have at most one forbidden color.
\end{proof}

In what follows, each node $v \in V \setminus \vpivot$ will remember its at most two forbidden colors due to neighbors in $\vpivot$. \Cref{def:strictH} also guarantees that all nodes in $V \setminus \vpivot$ will have at most one outgoing edge. So the nodes $w \in V \setminus \vpivot$ with no outgoing edge pick an arbitrary color that is not forbidden as their final color.

In order to properly color the nodes with exactly one outgoing edge, consider the following centralized procedure: Color the nodes one by one in a greedy manner starting from the highest layer and with an arbitrary order within one layer. Here, greedy means, that a node picks the smallest color that is not forbidden and not used by any of its already colored neighbors. This process computes a proper $3$-coloring as each node will have one color used by the neighbor along its outgoing edge, and at most one forbidden color. The output of a node $v\in V\setminus \vpivot$ in this centralized procedure only depends on the \emph{directed path of $v$} obtained by following outgoing edges starting at $v$. In the following lemma we show that this directed path cannot be too long.

\begin{lemma}
	The directed path of a node $v \in V \setminus \vpivot$ obtained by following outgoing edges starting at $v$ has length at most $O(\log n)$.
\end{lemma}
\begin{proof}
	Consider a directed edge $(u,w)$ in the directed path of $v$. If $\layer(u) = \layer(w)$, then $w$ cannot have an outgoing edge as it will have two neighbors in $V\setminus \vpivot$ with same or higher layer, violating \Cref{def:strictH}. In other words, if $\layer(u) = \layer(w)$, then the directed path of $v$ ends at $w$.
	
	Therefore, if we go along the directed path, the layer of the nodes either strictly increases or the path does not continue. Since there are $O(\log n)$ layers in the $H$-decomposition, the length of a directed path is at most $O(\log n)$. 
\end{proof}

In our \mpc algorithm, the idea is for each node to learn its $O(\log n)$ length directed path by performing graph exponentiation only along the directed edges. Since all nodes with no outgoing edges are already colored with their final color, consider performing the following \mpc algorithm only for nodes with one outgoing edge: Each node computes its final color after gathering its directed path by performing $O(\log \log n)$ graph exponentiation steps along directed edges.

\subsection{MIS and Maximal Matching} \label{sec:boundmismm}

The maximal independent set and maximal matching algorithms follow from \Cref{thm:bound3coloring}.

\begin{theorem} \label{thm:boundmis}
	There is a deterministic $O(\log \log n)$ round MIS algorithm for trees in the low-space \mpc model using $\widetilde{O}(n)$ words of global space.
\end{theorem}

\begin{proof}
	By \Cref{thm:bound3coloring}, we can color the tree with $3$ colors. For all colors $i$, perform the following. Nodes colored $i$ add themselves to the independent set, and all nodes adjacent to nodes colored $i$ remove themselves from the graph. Clearly this results in a maximal independent set in $O(1)$ rounds and the space requirements are satisfied. 
\end{proof}

\begin{theorem} \label{thm:boundmm}
	There is a deterministic $O(\log \log n)$ round maximal matching algorithm for trees in the low-space \mpc model using $\widetilde{O}(n)$ words of global space.
\end{theorem}

\begin{proof}
	By \Cref{thm:bound3coloring}, we can color the tree with 3 colors using a $H$-decomposition. Recall that in the decomposition, each node $v$ with $\layer(v) = i$ has at most two neighbors with layer at least $i$. Let us define the parent nodes of $v$. We orient an edge $\{u, v\}$ from $v$ to $u$ if (i) $u$ belongs to a strictly higher layer than $v$ or (ii) $u$ belongs to the same layer and has a higher ID. For all colors $i$, perform the following. Node $v$ colored $i$ proposes to its highest ID outgoing neighbor $u$, and $u$ accepts the proposal of the highest ID proposer. If $u$ accepts $v$'s proposal in which case the edge $\{u,v\}$ joins the matching. If $u$ rejects $v$'s proposal, it means that $u$ is matched with some other node and then we repeat the same procedure with $v$'s other possible out-neighbor. Note that when a node joins the matching, it prevents all other incident edges from joining the matching. As a result, all nodes colored $i$ have either joined the matching or they have no out-going edges. After iterating through all color classes, all nodes have either joined the matching or they have no incident edges, implying that all their original neighbors belong to the matching. This results in a maximal matching in $O(1)$ rounds and the space requirements are satisfied.
\end{proof}
\begin{proof}[Proof of \Cref{thm:bound3coloring}]
	Correctness follows from the fact that each node can recolor itself with its final color when seeing its whole directed path. The runtime follows from the fact that we only perform $O(\log \log n)$ graph exponentiation steps and color the directed paths.
	
	Let us analyze the space usage of our algorithm. Since the length of a directed path stored by each node during the algorithm is at most $O(\log n)$, we do not violate global space. Note that the sequential coloring of frozen layers does not require additional space. Notice that even though each node $v$ is the source of at most one request, multiple nodes may send a request to $v$. Hence, during graph exponentiation, node $v$ may have to communicate with a large number of nodes in lower layers. To mitigate this issue, we perform a load balancing process by sorting all the at most $n$ requests by the ID of their destination. This can be done deterministically in $O(1)$ rounds. Now, all the requests with destination $v$ lie in consecutive machines, and therefore, we can broadcast the response of $v$ to all these machines in $O(1)$ rounds by creating a constant depth broadcast tree on these machines. Therefore, each step of graph exponentiation can be done in $O(1)$ rounds, which leads to an overall running time of $O(\log \log n)$ rounds.
\end{proof}

\section{Coloring, MIS, Matching, and \texorpdfstring{$H$}{Lg}-decomposition with Optimal Space}
\label{sec:optimalSpace}
In this section we show how to obtain optimal global space by equipping the algorithm from \Cref{thm:boundmm,thm:boundmis,thm:bound3coloring} with suitable pre- and processing steps that free additional space.

\thmMainColoringMISMatching*
\begin{proof}
	We perform the standard procedure of iteratively putting in layer $i$ nodes of degree at most $2$ for $i = 1$ to $O(\log \log n)$. This removes $O(\poly \log n)$ fraction of the nodes since each iteration layers a constant fraction of the nodes. Therefore, the new number of nodes is $n' = n/\poly\log n$, and an \mpc algorithm using $\widetilde{O}(n')$ global space uses $O(n)$ words of global space.
	
	So we freeze these initial $O(\log \log n)$ layers obtain $G'$ remaining graph with $n'=n/\poly\log n$ nodes. Then we apply \Cref{thm:bound3coloring} to compute a $3$-coloring in $G'$ in $O(\log \log n)$ rounds and $O(n)$ global space. Finally we complete the solution on the nodes in the frozen layers one layer at a time taking an additional $O(\log \log n)$ rounds.
	
	The claim for MIS and maximal matching follows by the proofs of \Cref{thm:boundmis} and \Cref{thm:boundmm} respectively after computing the $H$-decomposition and the $3$-coloring.
\end{proof}

Using a similar preprocessing step, we can also show that a strict $H$-decomposition of \Cref{thm:strictHDecomppolylogMemory} can be computed with optimal global space.

\thmStrictHdecomp*
\begin{proof}
	Same as above, iteratively putting in layer $i$ the pivot nodes and nodes of degree $1$ as in \Cref{line:removePivotmain,line:removeDegOnemain} of \Cref{alg:strictHDecom} for $i = 1$ to $O(\log \log n)$. By \Cref{cor:strictH-single-step-progress} and using \Cref{lem:generalizedRakeCompress} with $x = 1$ and $\ell = 3$, we get that each iteration layers a constant fraction of nodes, which implies that $O(\poly \log n)$ fraction of the nodes are removed after $O(\log \log n)$ iterations.
	
	Now we have a partial strict $H$-decomposition if we assign layer $\infty$ to the remaining nodes. These nodes form a graph $G'$ with $n' = n/\poly\log n$ nodes, and so we can compute a strict $H$-decomposition on $G'$ using \Cref{alg:strictHDecom} in $O(\log \log n)$ rounds and $\widetilde{O}(n') = O(n)$ global space. Therefore, we have computed a strict $H$-decomposition of $G$ in $O(\log \log n)$ rounds and in $O(n)$ global space.
\end{proof}

\section{Balanced Exponentiation}
\label{sec:balanced_exponentiation}

Let $v \in F$ be a vertex of graph $F$ which is a forest. Recall that the (shortest) path between two nodes of a tree is unique. For all nodes $x\in N(v)$, define
\begin{equation*}
	F_{v \arr x}=\{w\in V(F)\mid \text{$x$  is contained in the shortest path from $v$ to $w$}\}
\end{equation*}

to be all nodes in the forest that are reachable from $v$ via $x$, including $x$. Also, let 
\begin{align*}
	F_{v \narr x} \coloneqq V(F) \setminus F_{v \arr x}~.
\end{align*}

Note that $v \in  F_{v \narr x}$ for any $x \in N(v)$. For brevity, we define $\eps \coloneqq \delta/8$, where $\delta$ comes from $O(n^\delta)$, the low-space \mpc bound per machine. Next, we define \emph{important} nodes, which will be central to our algorithm.

\begin{definition}[Important node] \label{def:importantNode}
	We call a node $v\in F$ \emph{important} if there exists a node $u\in N(v)$ such that $|F_{v \narr u}| \leq n^\eps$.
\end{definition}

The intuition for an important node is that there can at most be one neighbor towards which the graph contains many nodes. 

\begin{observation}
	\label{obs:importantOneHeavyDirection}
	Let $v\in F$ be an important node and $u \in N(v)$ a neighbor such that $|F_{v \narr u}| \leq n^\eps$. Then there is at most one $y\in N(v)$ with $|F_{v \arr y}|>n^{\eps}$, in which case $u$ is unique and $y=u$.
\end{observation}

\begin{proof}
	Let nodes $v$ and $u$ be as in the lemma statement. Assume for contradiction that there are two distinct nodes $y_1,y_2\in N(v)$ satisfying the lemma statement for $y$. W.l.o.g. it holds that $y_1\neq u$ and $y_1 \in N(v) \setminus \{u\}$. This implies $F_{v \arr y_1}\subseteq F_{v \narr u}$. We obtain $n^\eps < |F_{v \arr y_1}| \leq |F_{v \narr u}| \leq n^\eps$, a contradiction.
	
	If $|F_{v \arr y}|>n^{\eps}$ for $y \in N(v)$, the condition $|F_{v \narr u}| \leq n^{\eps}, u \in N(v)$ can hold for exactly one node, which is $y$. Hence, $y=u$.
\end{proof}

Let us give an informal version of the main theorem of the section. \Cref{lem:balancedExpInformal} is an informal version of \Cref{lem:balancedExp}, which may be independent interest. \Cref{lem:balancedExp} is standalone and can be used as a  blackbox in future works.

\begin{lemma}[Balanced Exponentiation, informal] \label{lem:balancedExpInformal} Let $0 < k \leq n^{\delta/8}$  be a parameter that may or may not be constant. There is a deterministic low-space \mpc algorithm that given an $n$-node forest $F$ uses $O(\log k)$ rounds in which every important node $v\in F$ discovers its $k$-hop neighborhood in every direction except at most one.
\end{lemma}

\subsection{Graph Definitions}
Given a node $v\in F$, we refer to each of its neighbors $x\in N(v)$ as a \emph{direction} with regard to $v$. 
For node $v$ and every node $w \in F$ define $r_v(w) \coloneqq x \in N(v)$ such that $w \in F_{v \arr x}$, i.e., $r_v(w)$ is the neighbor of $v$ which is on the unique path from $v$ to $w$ in $F$. 
Define $F^k_{v \arr x}$ and $F^k_{v \narr x}$ as $\{F_{v \arr x} \cap N^k(v)\}$ and $\{F_{v \narr x} \cap N^k(v)\}$, respectively.

\subsection{The Algorithm} 

At all times, every node $v \in F$ has some set of nodes $S_v$ in its memory, which we initialize to $N(v) \cup \{v\}$. Set $S_v$ is the node's view (or knowledge). Similarly to definitions $F_{v \arr x}$ and $F_{v \narr x}$, let us define $S_{v \arr x}=S_v\cap F_{v\arr x}$ and $S_{v \narr x} \coloneqq S_v \setminus S_{v \arr x}$.

\subsubsection{Exponentiation}

For a node $v$, value $k \in \mathbb{N}$, and any $X\subseteq N(v)$, an exponentiation operation, or exponentiating, is defined as computing
\begin{align*}
	\expo(X,k) \coloneqq \bigcup_{x \in X} \bigcup_{w \in S_{v \arr x}} S_{w \narr r_w(v)} \cap N^{k-d(v,w)}(w),
\end{align*}
where $N^{k-d(v,w)}(w)$ is the subgraph of radius $k-d(v,w)$ centered at $w$. Intuitively, the intersection in the definition is to prevent nodes from learning anything further that $k$-hops away. We say that a node $v$ exponentiates towards (or in the direction of) $x$ if $x \in N(v)$ and $v$ computes $\expo(X,k)$ with $x\in X$. An exponentiation operation is called \emph{uniform} if $X = N(v)$.

\subsubsection{Algorithm -- High Level Overview} \label{sec:algoOverview}

We prove the following lemma.

\begin{restatable}[Balanced exponentiation]{lemma}{lemBalancedExp}
	\label{lem:balancedExp}
	Let $0 < k \leq n^{\delta/8}$ be a parameter that may or may not be constant. Given an $n$-node forest $F$, there is a deterministic $O(\log k)$ time low-space \mpc algorithm after which the following holds. For every important node $v\in V(F)$ there is a node $z \in N(v)$ and a machine that for all $x \in N(v)\setminus z$ holds $N^{k}(v) \cap F_{v \arr x}$ in memory. The algorithm requires $O(n \cdot \poly(k))$ total space.
\end{restatable}

Even though \Cref{lem:balancedExp} applies to forests, we give an algorithm for trees. Since none of our arguments involve communication between nodes in disconnected components, we can execute our algorithm on every tree of a forest simultaneously in parallel. Henceforth, we assume the input graph is an $n$-node tree $G$.

\subparagraph*{Comparison with \cite{sodapaper}.} The authors of \cite{sodapaper} solve connectivity with the following core technique. Every node with a small $(\leq n^{\delta/8})$ subtree learns said subtree in $O(\log D)$ rounds, where $D=\diam(F)\leq n^{\delta/8}$. Their technique is intertwined with solving connectivity and cannot be black-boxed. Additionally, they have hardcoded parameter $D$ (more precisely $\hat{D}$ as the guess for $D$) into the algorithm, which makes it hard to modify.

More formally, their guarantees are as follows. Combining \cite[Lemmas 4.20, 4.24]{sodapaper} and deciphering their definitions gives that, if $\diam(F)=D\leq n^{\delta/8}$, every node $v \in G$ with a neighbor $u \in N(v)$ such that $|G_{v \narr u}|\leq n^{\delta/8}$ learns $G_{v \narr y}$ such that $|G_{v \narr y}| \leq n^{\delta/8}$. Observe that $y$ does not necessarily equal $u$. Their procedure requires $O(\log D)$ low-space \mpc rounds and $O(n \cdot D^3)$ total space.

Our result (\Cref{lem:balancedExp}) follows from extracting their technique from the connectivity solution, parametrizing it with $k$ (instead of hardcoding it to $D$), and obtaining the following claim. Every node $v \in G$ with a neighbor $u \in N(v)$ such that $|G_{v \narr u}|\leq n^{\delta/8}$ learns $G^k_{v \narr y} \coloneqq N^k(v) \cap G_{v \narr y}$. Similarly to \cite{sodapaper}, $y$ does not necessarily equal $u$.

Even though the guarantees of \Cref{lem:balancedExp} and the technique of \cite{sodapaper} look similar, there is a fundamental difference between them, resulting in the need for new analysis. The core idea in both procedures is to let nodes explore the graph such that in total, nodes do not learn too much, i.e., global space is not to large. In the global space analysis, roughly speaking, the idea is to imagine a rooting, and show that what nodes learn is in the order of the size of their subtree. In \cite{sodapaper}, a node $v$ is allowed to learn its $D$-hop neighborhood in \emph{any} $\deg(v)-1$ directions. This is fine in terms of global space, because this implies that node $v$ has succesfully learned its whole subtree. If we replace $D$ with a variable $k$, as in our generalized procedure, this claim no longer holds. If a node $v$ learns its $k$-hop neighborhood any $\deg(v)-1$  directions, we may have broken the global space analysis because $v$ could have learned too many nodes in the direction of the root.

\subparagraph*{High level overview.} 

The algorithm consists of $O(\log k)$ iterations, in each of which nodes perform a carefully designed exponentiation procedure. The aim is for all important nodes to become \comp (defined formally later), i.e., learn their $k$-hop neighborhood in all directions except at most one. 

The greatest challenge is abiding by the global space constraint, which is roughly speaking ensuring that $\sum_{v \in V} |S_v|$ is less than $O(n \cdot \poly(k))$. If there were no memory constraints and every node could perform a uniform exponentiation step in every iteration of the algorithm, i.e., execute $\expo(N(v),k)$, after $O(\log k)$ iterations \emph{all} nodes would learn their entire $k$-hop neighborhood and we would have achieved our goal. However, uniform exponentiation could potentially result in every node learning up to $n^\delta$ nodes, and as a result breaking exceed the allowed global space. Another difficulty is that, initially, no node $v \in G$ knows whether or not it is important, i.e., a node does not even know whether it has to learn most of its $k$-hop neighborhood or not. 

Hence, we need to steer the exponentiation with some strong invariant in order to abide by the memory constraints, yet we need to ensure that each important node still makes sufficient progress in learning large parts of its $k$-neighborhood. In order to achieve both goals at the same time, we perform careful probing for the number of nodes into all directions of node $v$ to determine in which directions we can safely exponentiate without using too much space (globally speaking). In the probing procedure $\pro$ (see \Cref{lem:mainProbing} and \Cref{sec:probingProcedure} for details), a node $v$ computes value $U_{v \arr x} = \sum_{w \in S_{v \arr x}} |S_{w \narr r_w(v)} \cap N^{k-d(v,w)}(w)|$  for every neighbor $x \in N(v)$ as an estimate for the number of nodes it may learn when exponentiating towards $x$. Note the similarities between the definitions of $U_{v \arr x}$ and \expo. The intuition is to prevent nodes from exponentiating in the directions that contain too many nodes by first perform probing. The estimate returned by the probing may be inaccurate and may contain doublecounting. Still, we ensure that every important node, in every iteration, makes progress in all but one direction. Next, we detail on the high level intuition for why this balanced exponentiation meets the global space requirements and why it makes sufficient progress. 

\subparagraph*{Global Space.}
By steering the exponentiation via the probing 
we uphold the following informal invariant at all times for each node $v\in G$:

\vspace{5mm}
\begin{center}
	\begin{minipage}{0.85\textwidth}
		\centering
		\begin{framed}
			\textit{For any node $y \in G$, the number of nodes that node $v$ sees in direction $y$ is bounded by $\poly(k) \cdot |G^k_{v \narr y}|$}
		\end{framed}
	\end{minipage}
\end{center}
\vspace{4mm}

Extending the procedures of \cite{sodapaper} such that all nodes learn sufficiently many nodes in sufficiently many directions while  meeting the invariant above is the main contribution of this section. In particular, our algorithm runs only for $O(\log k)$ rounds as opposed to $O(\log \diam(F))$ rounds as the algorithms in \cite{sodapaper}, complicating the proofs for the measure of progress. 
Anyhow, given the aforementioned invariant the global space bound can be deduced with the following observation from \cite{sodapaper}.

\begin{observation} \label{obs:globalmemory}
	If a tree is rooted, and every node $v$ counts the number of nodes in its subtree (subgraph rooted at $v$) up to distance $k$, all nodes have collectively counted  $n \cdot (k+1)$ nodes.
\end{observation}

\begin{proof}
	A node $x$ is counted, at most, only by the first $k$ nodes on the unique path from $x$ to the root. Hence, every node is overcounted at most $k+1$ times.
\end{proof}
Now, to prove that our algorithm meets the global space requirement of $O(n \cdot \poly(k))$ let us consider (for the sake of analysis) that the graph is rooted at an arbitrary node $y\in G$. Now, if every node would learn \emph{everything} (up to distance $k$) in all directions but in the direction of $y$, that is, node $v$ learns $G^k_{v\narr y}$, then \Cref{obs:globalmemory} ensures that the global space is upper bounded by $n \cdot (k+1)$ words. Now, the invariant shows that all information that nodes learn in the direction of $y$ is in the same order of magnitude, except for a $\poly (k)$ factor.

\subparagraph*{Local Space Per Machine.} Basically, we ensure that high degree nodes that are store (along with their edges) on multiple machines do not partake in the algorithm. For every other node $v \in G$ we ensure that set $S_v$ along with other information of $v$ is stored on the same machine as $v$. We achieve this by controlling the size of $S_v$: exponentiating in a direction $x$ should yield at most $n^{2\eps}$ nodes, which combined with the small degree of $v$ ensures our goal of $O(n^\delta)$ space per machine.

\subparagraph*{Measure of Progress.} We first show that for an important node $v\in G$, a large portion of its $k$-hop neighborhood consists of other important nodes. Then, we show that all important nodes exponentiate in all direction but one in every iteration, leading to exponential exploration speed in the distance between important nodes in $G$.

\subsubsection{Algorithm -- Detailed}
During the execution of the algorithm, each node $v\in G$ and all of its directions can be in different states, depending on the current knowledge of $v$, i.e., set $S_v$. More concretely, we say a direction  $x\in N(v)$ is in one of the following states (with regards to $v$).

\begin{itemize}
	\item \comp:  if $G^k_{v\arr x}\subseteq S_v$
	\item \blocked: If $\sum_{w \in S_{v \arr x}} \left|S_{w \narr r_w(v)} \cap N^{k-d(v,w)}(w)\right| > n^{\eps}$
	\item \act: if it is neither \comp nor \blocked
\end{itemize}

Throughout the algorithm, each node $v\in G$ maintains sets of \comp, \blocked, and \act  directions as  $C_v, B_v\subseteq N(v)$ and  $A_v =N(v)\setminus (C_v\cup B_v)$, respectively. 

So, intuitively a direction is \comp if node $v$ already knows everything in that direction that it wants to learn. The technical term in the definition of a \blocked direction stems, e.g., from the fact that our probing internally uses the left hand side term as its estimate. So, if that estimate is too large, we don't want node $v$ to be learning anything extra in that direction.

Similarly, node $v$ has the following states, that depend on the content of $S_v$. 

\begin{itemize}
	\item \comp: if its \comp for $\deg(v)-1$ directions
	\item \full: $|S_v| > n^{3\eps}$
	\item \act: if it is neither \comp nor \full
\end{itemize}

Note that an important node which becomes \comp knows everything that it wants to learn in order to satisfy \Cref{lem:balancedExp}. Also observe that any important node can never become \full because it can never learn $> n^{3\eps}$ nodes: (i) it holds that $|G_{v \narr u}|<n^\eps$ for some particular neighbor $u\in N(v)$ so $v$ cannot learn more than $n^\eps$ nodes in those directions and (ii) node $v$ is manually prevented from learning more than $n^{2\eps}$ nodes in any direction (including $u$). 

In \Cref{sec:probingProcedure}, we present the full procedure of \pro and prove the following lemma. 

\begin{restatable}[\pro]{lemma}{lemMainProbing}
	\label{lem:mainProbing}
	Let $0 < k \leq n^{\delta/8}$ be a parameter that may or may not be constant. In an arbitrary iteration $j$ of \Cref{alg:balancedExp}, $\pro(k,B_v)$ returns  a set of directions $\hds \subseteq N(v)$ and a set $X\subseteq  N(v)\setminus \hds$ of \emph{allowed directions} that ensure the following holds:
	
	\begin{enumerate}
		\item[(i)] $|\expo(x,k)| \leq k \cdot n^\eps$ for every $x\in X$, 
		\item[(ii)] $|\expo(x,k)| \leq k \cdot |G^k_{v \narr x}|$, for every $x\in X$, and
		\item[(iii)] $|X|\geq \deg(v)-1$ whenever $v$ is an important node. 
	\end{enumerate}
	
	The procedure can be implemented in $O(1)$ low-space \mpc rounds, in the same global space as in iteration $j$ of \Cref{alg:balancedExp}.
\end{restatable}

\Cref{alg:balancedExp} achieves the claims of \Cref{lem:balancedExp}.

\begin{algorithm}[t]
	\caption{Balanced Exponentiation}\label{alg:balancedExp}
	\begin{algorithmic}[1]
		
		\Function{\be}{$k$} 
		\State{Every node $v$ becomes \act. Initialize $S_v \larr N(v) \cup \{v\}$ and $C_v,B_v \larr \emptyset$.}
		\For{$j = 1,2,\ldots,O(\log k)$}
		\If{$|S_v| > n^{3\eps}$,}
		\State{$v$ becomes \full}
		\EndIf
		
		\State{\act node $v$ invokes $\hds, X  \larr \pro(k,B_v)$}
		
		\State{\act node $v$ updates $B_v \larr \hds$}
		
		\State{\act nodes $v$ updates $S_v\larr \{v\}\cup \expo(X,k)$} \phantomsection\label{l:exp}
		
		\If{$v \in S_{w \arr r_w(v)}$ but $w \not\in S_{v \arr r_v(w)}$ for some nodes $w \in G$} \phantomsection\label{l:forSymmetry}
		
		\If{adding nodes $w$ to $S_{v \arr r_v(w)}$ results in $|S_{v \arr r_v(w)}| \leq n^\eps$} \phantomsection\label{l:condition}
		
		\State{add nodes $w$ to $S_{v \arr r_v(w)}$} \phantomsection\label{l:add}
		\Else{~update $B_v \larr B_v \cup \{r_v(w)\}$} \phantomsection\label{l:block}
		
		\EndIf
		\EndIf
		
		\State{\act nodes update $C_v \larr \{x \in N(v) \mid G^k_{v \arr x} \subseteq S_{v \arr x}\}$}
		
		\If{$|C_v| \geq \deg(v)-1$}
		\State{node $v$ becomes \comp}
		\EndIf
		
		\If{an direction $b$ has been exponentiated towards 3 times after $J= \lceil  \log_{5/6} (5/k) \rceil$ iterations without $b$ becoming \comp} \phantomsection\label{l:blockIfNotComp}
		\State{update $B_v \larr B_v \cup \{b\}$}
		\EndIf
		\EndFor

		\State{\Return{$\{S_v \mid v \in G \}$}}	
		\EndFunction

	\end{algorithmic}
\end{algorithm}

\Cref{l:blockIfNotComp} is due to a technicality in the measure of progress (\Cref{lem:importantKnow}). The value $J$ is deduced in \Cref{lem:heavyDirClose}, and the fact that direction $b$ fulfills the requirements of a blocked direction is shown in \Cref{obs:canBlock}.

At the end of \Cref{sec:localMemory}, we prove the local space bounds of \Cref{lem:balancedExp}. In \Cref{sec:progressAndCorrectness}, we discuss the measure of progress and prove the correctness of \Cref{lem:balancedExp} at the end of said section. At the end of \Cref{sec:globalMemory}, we prove the global space bounds of \Cref{lem:balancedExp}. The subroutine \pro of \Cref{alg:balancedExp} is introduced and proved in \Cref{sec:probingProcedure}. The missing \mpc details of \Cref{alg:balancedExp} that do not appear in \Cref{sec:localMemory,sec:progressAndCorrectness,sec:globalMemory,sec:probingProcedure} are deferred to \Cref{sec:MPCdetails}. \Cref{sec:proofOfLemma} is dedicated to proving \Cref{lem:subsets}, which gives the necessary sets to construct an $H$-decomposition in \Cref{sec:subtreeRC}.

Let us start with following lemma, which is crucial both for space bounds and the measure of progress of the algorithm.

\begin{lemma}[Properties of \Cref{alg:balancedExp}] \label{lem:activesymmetry}
	Consider an arbitrary iteration of the algorithm. For node $v\in V$ let $S_v$ ($S'_v$) denote the node's knowledge at the beginning (end) of the iteration. For all nodes $v \in G,~x \in N(v)$ the following holds.
	\begin{enumerate}
		\item At the end of every iteration, if direction $x$ is \act for node $v$ and if $v \in S_w$ for some node $w \in G_{v \arr x}$, it holds that $w \in S_{v \arr x}$.
		\item It always holds that $S_v\subseteq S_v'$.
		\item It always holds that $S'_v \subseteq N^k(v)$.
	\end{enumerate}
\end{lemma}

\begin{proof}[Proof of 1.]
	Let us prove the claim for an arbitrary iteration $j$. After \Cref{l:exp} in any iteration $i<j$, if some node $w \in G$ added $v$ to $S_w$ without $v$ adding $w$ to $S_v$, we enter \Cref{l:forSymmetry}. Because direction $x$ is \act in iteration $j$, it was not blocked in iteration $i$ in \Cref{l:block}. Hence, \Cref{l:forSymmetry} always leads to \Cref{l:add} and $w$ will be added to $S_v$.
\end{proof}

\begin{proof}[Proof of 2.]
	Set $S_v$ gets updated only in \Cref{l:exp,l:add}. In the latter, it clearly cannot lose elements. For the former, let us recall what $\expo(x,k)$ returns for a direction $x \in N(v)$:
	\begin{align*}
		\bigcup_{w \in S_{v \arr x}} S_{w \narr r_w(v)} \cap N^{k-d(v,w)}(w)~.
	\end{align*}
	
	Because $w \in S_{w \narr r_w(v)}$, the set above includes everything in $S_{v \arr x}$. Because 
	\begin{align*}
		\expo(X,k)= \bigcup_{x \in X} \expo(x,k) \supseteq S_v \setminus \{v\}
	\end{align*} 
	
	and the union in \Cref{l:exp} includes $v$ itself, the claim holds.
\end{proof}

\begin{proof}[Proof of 3.]
	Set $S_v$ gets updated only in \Cref{l:exp,l:add}. By the definition of $\expo$, \Cref{l:exp} cannot result in $S_v$ containing some node $w \not\in N^k(v)$. In \Cref{l:add} node $v$ adds node $w$ to $S_v$ only if $v \in S_w$, which already implies that $d(v,w)\leq k$.
\end{proof}

When a node $v \in G$ stores $w$ in $S_v$, in practice it means that $v$ stores the tuple 
\begin{align}
	t_v(w) = (v,r_v(w),w,r_w(v),d(v,w))
\end{align} 

in $S_v$. The following lemma ensures that a tuple stored in $S_v$ is always valid. Observe that in a tree, a valid tuple is unique.

\begin{lemma} \label{lem:validTuple}
	Let $v \in G$. For every $w \in S_v$, tuple $t_v(w)$ is valid.
\end{lemma}

\begin{proof}
	When initializing $S_v \larr N(v) \cup \{v\}$, we add $(v,w,w,v,1)$ for $w \in N(v)$ and $(v,v,v,v,0)$ for $v$, which are valid. Consider nodes $x,y,z$ such that $y \in S_x$, $z \in S_y$, $z \not\in S_x$, and $x$ performing an exponentiation step towards $r_x(y)$. Let us show how $v$ can compute a valid tuple $t_x(z)$ by only using its own tuple $t_x(y)$ and $y$'s tuple $t_y(z)$. 
	\begin{align*}
		\left.
		\begin{cases}
			~t_x(y)= (x,r_x(y),y,r_y(x),d(x,y)) \\
			~t_y(z)= (y,r_y(z),z,r_z(y),d(y,z))
		\end{cases}
		\hspace{-3mm}
		\right\} ~\xrightarrow[~v~]{}~t_x(z)= (x,r_x(y),z,r_z(y),d(x,y)+d(y,z))
	\end{align*}
	
	The validity of the above is based on the fact that in a tree, a path between $x$ and $z$ is unique. In \Cref{l:add} of the algorithm, node $v$ may add $w$ to $S_v$, if $v \in S_w$. This operation is even simpler that the previous one.
	\begin{equation*}
		\left.
		\begin{cases}
			~t_w(v)=(w,r_w(v),v,r_v(w),d(w,v))\\
		\end{cases}
		\hspace{-3mm}
		\right\} ~\xrightarrow[~v~]{}~ t_v(w)=(v,r_v(w),w,r_w(v),d(w,v)) \qedhere
	\end{equation*}
	
\end{proof}

The following lemma is key in order for nodes to correctly deduce their states and the states of their directions in \Cref{lem:learnLabel}.

\begin{lemma} \label{lem:learnInducedGraph}
	A node $v \in G$ can determine the graph $G[S_v]$ in $O(1)$. 
\end{lemma}

\begin{proof}
	A node $v$ storing another node $w \in S_v$ in memory actually stores tuple $t_v(w)$. There is an edge $\{a,b\}$ in $G[S_v]$ if and only if there exist tuples $t_v(a)$ and $t_v(b)$ in $S_v$ such that
	\begin{align*}
		t_v(a) &= (v,\ast,a,b,\ast) \\
		t_v(b) &= (v,\ast,b,\ast,\ast)~,
	\end{align*}
	
	where $\ast$ means that the entry can be arbitrary. The tuples above imply that nodes $a,b \in G[S_v]$ and that the edge adjacent to $a$ that is on the unique path from $a$ to $v$ is also adjacent to $b$. This covers all edges in the node-induced subgraph $G[S_v]$.
\end{proof}

Hence, we get the following lemma. 

\begin{lemma} \label{lem:learnLabel}
	A node can determine if a direction $x \in N(v)$ is \comp. A node can determine whether or not it is \full or \comp. This takes $O(1)$ rounds.
\end{lemma}

\begin{proof}
	A direction $x$ is \comp if $G^k_{v\arr x}\subseteq S_v$. By \Cref{lem:learnInducedGraph}, node $v$ can determine the graph $G[S_v]$. Furthermore, node $v$ can compute the degrees (in $G[S_v]$) of all nodes in $S_v$. Node $v$ can also query the degrees (in $G$) of all nodes $S_v$. If the degrees in $G[S_v]$ matches the degrees in $G$ of all nodes within distance $k-1$, it holds that $G^k_{v\arr x}\subseteq S_v$ and $x$ is \comp. If $\deg(v)-1$ directions are \comp, so is $v$. Node $v$ is full if $|S_v|>n^{3\eps}$. Hence, $v$ just needs to count the number of elements in $S_v$.
\end{proof}

\subsection{Local Space Bound} \label{sec:localMemory}
Informally, \Cref{lem:notLearnTooMuch} (presented below) shows that a node $v \in G$ cannot learn too many nodes in a single direction. Note that this lemma holds regardless of the number of iterations for which we execute the loop in \Cref{alg:balancedExp}. In fact, all of our space bounds would even hold if we would  continue running the process forever; instead, at some point nodes would simply turn inactive (more precisely either \full or \comp).

\begin{lemma} \label{lem:notLearnTooMuch}
	For any nodes $v \in G,~x \in N(v)$ at the end of an arbitrary iteration it holds that $|S_{v \arr x}|\leq k\cdot n^{\eps}$. 
\end{lemma}

\begin{proof}
	Let node $v$ be as in the lemma statement. Proof by induction. By initialization, $|S_{v \arr x}|=1$. Assume that the claim holds in iteration $j-1$. In iteration $j$, set $S_{v \arr x}$ grows only in \Cref{l:exp,l:add}. If \Cref{l:exp} is executed, it holds that $S_{v \arr x} = \{v\} \cup \expo(x,k)$ and the claim holds due to guarantee (i) of \Cref{lem:mainProbing}. The growth due in \Cref{l:add} is capped manually in the condition of \Cref{l:condition}, which ensures that $|S_{v \arr x}| \leq n^\eps$.
\end{proof}

\begin{lemma}[Local space for important nodes] \label{lem:localMemoryImportant} 
	Consider an important node $v \in G$ and its neighbor $u \in N(v)$ for which it holds that $|G_{v \narr u}| \leq n^\eps$. Node $v$ never becomes \full, i.e., $|S_{v}| < n^{3\eps}$ at all times.
\end{lemma}

\begin{proof}
	Let node $v$ be as in the lemma statement. Node $v$ cannot become \full due to initialization $S_v \larr N(v) \cup \{v\}$, since for an important node we have $|N(v)| +1 \leq |G_{v \narr u}| +1 \leq  n^\eps +1 < n^{3\eps}$. During the execution of the algorithm it holds that
	\begin{align*}
		|S_v| &= |S_{v \arr u}| +  |S_{v \narr u}| \stackrel{*}{\leq} k \cdot n^{\eps} + |G_{v \narr u}| \stackrel{**}{\leq} k \cdot n^{\eps} + n^\eps < n^{3\eps} ~.
	\end{align*}
	
	In $*$ we use \Cref{lem:notLearnTooMuch} for direction $u$, and the trivial upper bound of $|G_{v \narr u}|$ on $|S_{v \narr u}|$.
	In $**$ we apply lemma assumption of node $v$ being important. In the latter inequalities we use that $k< n^{\delta/8}$ and $\eps = n^{\delta/8}$. Hence, node $v$ never becomes \full.
\end{proof}

\begin{lemma}[Local space for non-important nodes] \label{lem:localMemoryNonImportant} 
	Consider a non-important node $v \in G$ with $\deg(v) \leq n^{3\eps}$. It always holds that $|S_{v}| < n^{6\eps}$.
\end{lemma}

\begin{proof}
	Let node $v$ be as in the lemma statement. During the algorithm, it holds that
	\begin{align*}
		|S_v| &= \sum_{x \in N(v)} |S_{v \arr x}| + 1 \stackrel{*}{\leq} \deg(v) \cdot n^{2\eps} + 1 \stackrel{**}{\leq} n^{3\eps} \cdot n^{2\eps} + 1 < n^{6\eps}~.
	\end{align*}
	In $*$ we use \Cref{lem:notLearnTooMuch} for every direction, and in $**$ we use the lemma assumption.
\end{proof}

\begin{proof}[Proof of \Cref{lem:balancedExp}: Local space]
	The local space bound of low-space \mpc is that every machine uses at most $O(n^\delta)$ space. What we implicitly assume during \Cref{alg:balancedExp} is that whenever node $v$ does anything meaningful, the machine $M(v)$ storing $v$ can do it using local computation in $O(1)$ time because it has all the data it needs in memory, e.g., that $|S_v|$ is fully contained in the memory of $M(v)$. Let us analyze what data every node needs in order to perform \Cref{alg:balancedExp} and show that this is at most $O(n^\delta)$ words.
	
	Let $v \in G$. If $\deg(v)>n^{3\eps}$, node $v$ is not \act because $|S_v| = \deg(v)+1 > n^{3\eps}$ and hence $v$ does not partake in the algorithm except for answering queries from other nodes, which is handled in \Cref{sec:MPCdetails}. If $\deg(v) \leq n^{3\eps}$, then by \Cref{lem:localMemoryImportant,lem:localMemoryNonImportant} $|S_v|<n^{6\eps}$ holds. For every element $w \in S_v$, node $v$ stores a tuple $(v,r_v(w),w,r_w(v),d_G(v,w))$, which takes $5$ words. Additionally, node $v$ stores its state, and the state of all its directions, which are bounded by $\deg(v)$. Storing a state takes one word. In total, the local space required by a node $v$ comprises of $5|S_v|+1+\deg(v)<6n^\delta=O(n^\delta)$ words.
\end{proof}

\subsection{Measure of Progress and Correctness} \label{sec:progressAndCorrectness}

For the sake of analysis, we define a virtual graph $G_j$ on node set $V$ for each iteration $j$ of the algorithm. Its edges depend on the nodes' knowledge and hence the graph changes with every iteration.

\subparagraph*{Virtual Graph $G_j$ (for analysis only).}
Let $G=(V,E)$ be the input graph, fix an arbitrary iteration $j$ of the algorithm and let $S_v$ denote the knowledge of node $v$ at the beginning of the iteration. Then $G_j=(V,E_j)$ is defined as follows.
\begin{align*}
	E_j = \{ ~&\{v,w\} \mid v,w\in V \text{and } w \in S_v \text{ or }  v \in S_w \}
\end{align*}

Note that $E_j$ also contains the edge $\{v,w\}$ if only one of its endpoints knows about the other node. 
Recall, that by definition for any important node $v$ there exists at least one neighbor $u \in N(v)$ such that $|G_{v \narr u}| \leq n^\eps$ holds. 

The following restatement of \Cref{lem:activesymmetry} is to remind the reader of this crucial property. 

\begin{restatement}[\Cref{lem:activesymmetry} part 1] \label{res:symmetry}
	At the end of every iteration, if direction $x$ is \act for node $v$ and if $v \in S_w$ for some node $w \in G_{v \arr x}$, it holds that $w \in S_{v \arr x}$.
\end{restatement}

The following lemma shows many useful properties of important nodes. In particular, how many \blocked directions they can have, and in which directions will they be forced to exponentiate (hence ensuring progress).

\begin{lemma} \label{lem:importantNodeProperty}
	Consider an important node $v\in G$ and a neighbor $u \in N(v)$ such that $|G_{v \narr u}| \leq n^\eps$.
	
	\begin{enumerate}
		\item (at most one \blocked iteration) Node $v$ can have at most one \blocked direction. If $v$ has a \blocked direction, then the choice of $u$ is unique and $u$ has to be the \blocked direction. 
		\item  Consider an iteration $j$ in which some direction $b \in N(v)$ is \blocked for node $v$. Then, in all iterations $j'\geq j$ node $v$ will exponentiate towards $N(v)\setminus b$.
		\item (symmetry) If there is an edge $\{v,w'\} \in E_j$ in a direction that is included in $N(v)\setminus u$, it holds that $w' \in S_v$.
	\end{enumerate}
\end{lemma}

\begin{proof}[Proof of 1.]
	By \Cref{obs:importantOneHeavyDirection} direction $u$ is the only one for which it may hold that $|G_{v \arr y}| > n^\eps$. Hence, it is the only one for which it may hold that $|G^k_{v \arr y}| > n^\eps$, implying that only direction $u$ can be \blocked. 
\end{proof}

\begin{proof}[Proof of 2.]
	By part 1, node $v$ can only have one \blocked direction. If a direction $b$ is \blocked, then $\hds=\{b\}$ and by \Cref{lem:mainProbing}, $X = N(v) \setminus b$. The update step in \Cref{l:exp} is $S_v\larr \{v\}\cup \expo(X,k)$, concluding the proof.
\end{proof}

\begin{proof}[Proof of 3.]
	If $v \not\in S_{w'}$, it must be that $w' \in S_v$ by the definition of $E_j$. So we can assume that $v \in S_{w'}$. By part 1, a direction $x \in N(v) \setminus u$ cannot be \blocked. Hence $x$ is either \act or \comp. If $x$ is \act, the claim holds by \Cref{lem:activesymmetry} part 1. 
	
	The fact that $v \in S_{w'}$ implies that $w' \in G^k_{v \arr x}$ by \Cref{lem:activesymmetry} part 3. If $x$ is \comp, then by definition it must be that $G^k_{v \arr x} \subseteq S_{v \arr x}$, and the claim holds.
\end{proof}

\Cref{lem:progress} is perhaps the most technical lemma of the section. Informally, it shows that in every iteration of the algorithm, important nodes make sufficient progress 
in learning what they need to learn. For the technical parts of the lemma, we define what is a path and what is a subpath in $G$.

A path of length $k$ in a graph $G=(V,E)$ is a sequence $v_1,\ldots,v_k$ of nodes such that $\{v_i,v_{i+1}\}\in E$ for all $1\leq i\leq k-1$. 
A path $P'=\left(v_{a_1},\ldots,v_{a_{k'}}\right)$ is a \emph{subpath} of a path  $P=(v_1,\ldots,v_k)$, if $a_i< a_{i+1}$ for $1\leq i \leq k'-1$. 

\begin{lemma} \label{lem:progress} 
	Consider an important node $v \in G$, an arbitrary neighbor $u \in N(v)$ satisfying $|G_{v \narr u}| \leq n^\eps$, a node $w \in G_{v \narr u}$, and an arbitrary iteration $j$ of \Cref{alg:balancedExp}. Let $P^j_{vw}$ be a path in $G_j$ between $v$ and a node $w$ satisfying the following. 
	
	\begin{enumerate}
		\item $P^j_{vw}$ is a subpath of the unique path between $v$ and $w$ in $G$. 
		\item $|P^{j}_{vw}| \geq 6$, i.e., $P^{j}_{vw}$ consists of $\geq 6$ nodes. 
	\end{enumerate} 
	
	Then then there exists a path $P^{j+1}_{vw}$ in $G_{j+1}$ between $v$ and $w$ satisfying $|P^{j+1}_{vw}| \leq \lceil 5/6 \cdot |P^j_{vw}| \rceil$.
\end{lemma}

\begin{proof}
	Let $v,u,$ and $w$ be as in the lemma statement. Observe that since $w \in G_{v \narr u}$, by property 1 of the lemma statement, it holds that every node $w' \in P^j_{vw}$ is also in $G_{v \narr u}$. Next, observe that every node $w' \in P^j_{vw}$ is important because $w' \in G_{v \narr u}$: $|G_{v \narr u}| \leq n^{\eps}$ implies that $|G_{w' \narr r_{w'}(v)}| \leq n^{\eps}$ for any node $w' \in G_{v \narr u}$, i.e., nodes $w'$ are important because they have a neighbor $r_{w'}(v)$ such that \Cref{def:importantNode} holds.
	
	Consider any subpath $P = \{x_1,x_2,x_3,x_4,x_5,x_6\} \subseteq P^j_{vw}$ of length $6$ with $d_G(v,x_1)<\ldots <d_G(v,x_6)$. For $1\leq i\leq 6$, let $S_{x_i}$ ($S'_{x_i}$) be the memory of node $x_i$ at the start (end) of iteration $j$. Consider nodes $x_3$ and $x_4$. We split into two cases, either (i) neither have a \blocked direction or (ii) at least one of them has a \blocked direction.
	
	For the following, observe that because $P$ is a subpath of $P^j_{vw}$ ($x_1$ is closer to $v$ than $w$) and $|G_{v \narr u}| \leq n^\eps$, $x_2$ $(x_3)$ fulfills the conditions of $u$ for $x_3$ $(x_4)$ in \Cref{lem:importantNodeProperty} part 3 and it holds that $x_4 \in S_{x_3}$ $(x_5 \in S_{x_4})$.
	
	\begin{itemize}
		\item[(i)] If $r_{x_4}(x_3)$ $(r_{x_3}(x_2))$ is \comp for node $x_4$ $(x_3)$, there already exists a path between $x_1$ and $x_6$ of length $<6$ because $x_1 \in S_{x_4}$ $(x_1 \in S_{x_3})$. 
		
		Hence, because node $x_4$ $(x_3)$ has not blocked $r_{x_4}(x_3)$ $(r_{x_3}(x_2))$, the direction is \act, which by \Cref{res:symmetry} implies that $x_3 \in S_{x_4}$ $(x_2 \in S_{x_3})$. By guarantee (iii) of \Cref{lem:mainProbing}, node $x_4$ exponentiates in $\deg(x_4)-1$ directions, i.e., towards $r_{x_4}(x_3)$ or $r_{x_4}(x_5)$, which will result in $x_2 \in S'_{x_4}$ or $x_6 \in S'_{x_4}$. In either case, there will exists a path between $x_1$ and $x_6$ in $G_{j+1}$ of length $5$.
		
		\item[(ii)] 
		First assume that node $x_3$ has a \blocked direction which has to be $r_{x_3}(u)\neq r_{x_3}(x_4)$ by \Cref{lem:importantNodeProperty} part 1; we first ignore whether $x_4$ has a \blocked direction or not. Because $x_3$ is an important node, by \Cref{lem:importantNodeProperty} part 2, it exponentiates towards $N(x_3) \setminus r_{x_3}(u)$. Since $r_{x_3}(x_4) \in N(x_3) \setminus r_{x_3}(u)$, this results in $x_5 \in S'_{x_3}$ (observe that $x_5\in S_{x_4}$ even if $x_4$ has a \blocked direction). Hence there will exists a path between $x_1$ and $x_6$ in $G_{j+1}$ of length $5$ in this case. 
		
		If $x_3$ has no direction \blocked, but $x_4$ has a direction \blocked, the proof is identical to the previous case with all indices shifted by one. \qedhere
	\end{itemize}
\end{proof}

We obtain the following lemma by iterating \Cref{lem:progress}.

\begin{lemma} \label{lem:heavyDirClose}
	Consider an important node $v \in G$ and a neighbor $u \in N(v)$ for which it holds that $|G^k_{v \narr u}| \leq n^\eps$. After $J=\log_{5/6} (5/k)$ iterations, it holds that for any $w \in G^k_{v \narr u}$ we have $d_{G_J}(v,w)<5$.
\end{lemma}

\begin{proof}
	At the start of first iteration of the algorithm, the unique path from $v$ to $w$ fulfills the conditions of $P^0_{vw}$ in \Cref{lem:progress}, so hence there exists path $P^1_{vw}$ such that $|P^{1}_{vw}| \leq \lceil 5/6 \cdot |P^0_{vw}| \rceil$. Path $P^1_{vw}$ also fulfills the conditions of \Cref{lem:progress}. By iteratively applying the result of \Cref{lem:progress} we arrive at a path $P^J_{vw}$ for which the second condition of \Cref{lem:progress} does not hold, i.e., $|P^J_{vw}|<6$. This is equivalent to $d_{G_J}(v,w)<5$, proving the claim (recall that the length of a path is defined as the number of nodes, and distance as the number of edges). Hence $J = \lceil g \rceil = O(\log k)$ from
	\begin{align*}
		(5/6)^g \cdot k &= 5 \\
		g &= \log_{5/6} (5/k) \qedhere
	\end{align*}
\end{proof}

The following lemma gives us 1-hop progress, which is needed to finish of the correctness proof, i.e., show that important nodes end up learning everything they need to learn. 

\begin{lemma} \label{lem:progressOne}
	Consider an arbitrary iteration $j$ where $d_{G_j}(v,x)=d_{G_j}(x,w)=1$ for distinct nodes $v,x,w \in G$ where $x\in G_{v\arr r_v(w)}$ and $x\in G_{w\arr r_w(v)}$. Additionally, $x \in S_v$ and $w \in S_x$. If $v$ updates $S_v \larr \{v\} \cup \expo(X,k)$ such that $x \in X$, we have  $d_{G_{j+1}}(v,w)=1$. 
\end{lemma}

\begin{proof}
	Note that $d_{G_j}(v,w)\leq 2$. 
	If $d_{G_j}(v,w)=1$ the distance cannot increase in the next iteration (see part 2 of \Cref{lem:activesymmetry}). 
	Now consider $d_{G_j}(v,w)=2$. As $x\in G_{v\arr r_v(w)}$ and $w\in S_x$ we obtain $w\in \expo(x,k)$ (executed by node $v$). Hence, $v$ adds $w$ to $S_v$ when updating $S_v \larr \{v\} \cup \expo(X,k)$ with $x\in X$ and we obtain $d_{G_{j+1}}(v,w)=1$.
\end{proof}

\Cref{lem:importantKnow} shows that after \Cref{lem:heavyDirClose} brought important nodes close (in $G_j$) to the nodes they want to learn, we can apply the 1-hop progress of \Cref{lem:progressOne} to make important nodes \comp. Recall, that a node $v$ is \comp, if $S_{v \arr x} \subseteq G^k_{v \arr x}$ holds for at least $\deg(v)-1$ neighbors $x \in N(v)$.

\begin{lemma} \label{lem:importantKnow}
	Consider an important node $v \in G$. After $O(\log k)$ iterations, node $v$ becomes \comp. 
\end{lemma}

\begin{proof}
	Let $v$ be as in the lemma statement and pick any node $u \in N(v)$ such that $|G_{v \narr u}| \leq n^\eps$ holds. Node $u$ is arbitrary but remains fixed throughout the proof.  By \Cref{lem:heavyDirClose}, we know that after $J=O(\log k)$ iterations, it holds that for any $w \in G^k_{v \narr u}$ we have $d_{G_J}(v,w)<5$. Let us consider the following $>J$ iterations, and show that there are only a constant number before $v$ becomes \comp. Recall that \Cref{lem:importantNodeProperty} part 3 holds for $v$ and all nodes in $w \in G^k_{v \narr u}$.
	
	By \Cref{lem:importantNodeProperty} part 1, node $v$ can only block direction $u$. Hence, directions $N(v)\setminus u$ are always either \act or \comp. There are three possible cases to analyze:
	
	\begin{enumerate}
		\item Direction $u$ is \blocked.
		\item Direction $u$ is \comp.
		\item Direction $u$ is \act.
	\end{enumerate}
	
	Observe that after node $v$ exponentiates in a direction $x \in N(v) \setminus u$ at most 3 times, that direction will become \comp: direction $x$ cannot be \blocked by \Cref{lem:importantNodeProperty} part 1, the distance between $v$ and nodes $w \in G^k_{v \arr x}$ is < 5, and by \Cref{lem:progressOne} that distance reduces by one in every exponentiation step.
	
	\begin{enumerate}
		\item Since $v$ is important and $u$ is \blocked, by \Cref{lem:importantNodeProperty} part 2, $v$ exponentiates in directions $N(v) \setminus u$ for the remainder of the algorithm. Node $v$ will become \comp after at most 3 iterations.
		
		\item By guarantee (iii) of \Cref{lem:mainProbing} it holds that $|X|=\deg(v)-1$ and $v$ will exponentiate in all directions $N(v) \setminus u$ but one. We claim that after at most 5 iterations $|N(v) \setminus u|-1$ directions will become \comp and so will $v$. 
		
		In order to show that after $5$ further iterations all $N(v) \setminus u$ directions except for one are \comp consider the following experiment. We have a bucket for each direction (out of directions $N(v) \setminus u$), and in each iteration we place a token into the bucket towards which $v$ exponentiated. After 5 iterations, there can be at most one direction/bucket with $<3$ tokens an hence all but one directions will become \comp.

		\item Consider the following 3 iterations where node $v$ exponentiates towards $u$. If $u$ becomes \comp, we are in case 2 and the claim holds after at most 5 iterations. Otherwise, $u$ becomes \blocked manually by the algorithm in \Cref{l:blockIfNotComp} and we are in case 1 and the claim holds after at most 3 iterations. 
		
		In the iterations where node $v$ does not exponentiate towards $u$, it exponentiates in all other $N(v) \setminus u$ directions because of guarantee (iii) in \Cref{lem:mainProbing}. After at most 3 such iterations directions $N(v) \setminus u$ will become \comp and so will $v$. Hence, the claim will hold after a constant number of iterations. 
		\qedhere
	\end{enumerate}
\end{proof}

The following observation is to ensure that if a direction $b$ is manually blocked by the algorithm in \Cref{l:blockIfNotComp}, it indeed holds that $|G^k_{v \arr b}|>n^\eps$.

\begin{observation} \label{obs:canBlock}
	Consider any node $v \in G$ and $b \in N(v)$. If after $\log_{5/6} (5/k)$ iterations node $v$ exponentiates towards $b$ at most 3 times and it does not become \comp, it must be that $|G^k_{v \arr b}|>n^\eps$.
\end{observation}

\begin{proof}
	Assume that $|G^k_{v \arr b}| \leq n^\eps$ for contradiction. Observe that the proof of \Cref{lem:progress} is independent for every direction. Hence, the claim in \Cref{lem:heavyDirClose} actually holds for direction $b$, i.e., after $\log_{5/6} (5/k)$ iterations, for any $w \in G^k_{v \arr b}$, we have $d_{G_J}(v,w)<5$. So if node $v$ exponentiates towards $b$ at most 3 times, $b$ direction will become \comp (a contradiction): direction $b$ cannot be \blocked by \Cref{lem:importantNodeProperty} part 1, the distance between $v$ and nodes $w \in G^k_{v \arr x}$ is < 5, and by \Cref{lem:progressOne} that distance reduces by one in every exponentiation step. 
\end{proof}

\begin{proof}[Proof of \Cref{lem:balancedExp}: Correctness]
	By \Cref{lem:importantKnow}, after $O(\log k)$ iterations, all important nodes become \comp. By definition of \comp, for every important node $v$ it holds that $G^k_{v \arr x} = N^k(v) \cap G_{v \arr x} \subseteq S_v$ for $\deg(v)-1$ distinct directions $x$. This fulfills the requirements of \Cref{lem:balancedExp}.
\end{proof}

\subsection{Global Space Bound} \label{sec:globalMemory}
Our space bounds hold regardless of the number of iterations. 

\begin{lemma}[Global space] \label{lem:globalMemoryBound}
	After every iteration it holds that $\sum_{v \in G} |S_v| = O(n \cdot \poly(k))$. 
\end{lemma}

\begin{proof}
	At the start of the algorithm, due to the initialization of $S_v$ for every node in $G$, it holds that $\sum_{v \in G}|S_v| = n +2m = 3n$. 
	Consider the end of an arbitrary iteration. We will bound $\sum_{v\in V}|S_v|$ at the end of the iteration. For the sake of analysis, let us root the tree at an arbitrary node $y \in G$. We can partition $S_v$ into $S_{v \narr r_v(y)}$ and $S_{v \arr r_v(y)}$ (note that $v\in S_{v \narr r_v(y)}$). Let us split set $S_{v \arr r_v(y)}$ into two parts as follows.
	
	\begin{itemize}
		\item $S_{v \arr r_v(y)}^{\exp}$: nodes that have been added to $S_{v \arr r_v(y)}$ the last time $v$ performed $\expo(X,k)$ such that $y \in X$ (\Cref{l:exp})
		\item $S_{v \arr r_v(y)}^{\text{manual}}$: nodes that have been added to $S_{v \arr r_v(y)}$ by $v$ manually in \Cref{l:add}, to ensure the symmetry condition (\Cref{lem:activesymmetry} part 1), after the last time $v$ performed $\expo(X,k)$ such that $y \in X$.
	\end{itemize}
	 
	With these definitions, we can write that
	\begin{align}
		\label{eqn:Sv}
		\sum_{v \in G} |S_v| =	\sum_{v \in G} \left( |S_{v \narr r_v(y)}| + |S_{v \arr r_v(y)}^{\exp}| + |S_{v \arr r_v(y)}^{\text{manual}}| \right)~.
	\end{align}
	
	We first bound the term $\sum_{v\in V}|S_{v \arr r_v(y)}^{\text{manual}}|$ in terms of the other terms. Observe that for every node $w \in S_{v \arr r_v(y)}^{\text{manual}}$, there exists node $v \in S_w \setminus S_{w \arr r_w(v)}^{\text{manual}}$. Hence, we obtain 
	\begin{align*}
		\sum_{v\in V}|S_{v \arr r_v(y)}^{\text{manual}}|\leq \sum_{w \in G}|S_w \setminus S_{w \arr r_w(v)}^{\text{manual}}|  = \sum_{v \in G} \left( |S_{v \narr r_v(y)}| + |S_{v \arr r_v(y)}^{\exp}|\right)~. 
	\end{align*}
	
	Pluggin this into \Cref{eqn:Sv} we obtain
	\begin{align*}
		\sum_{v \in G} |S_v| \leq 	2\sum_{v \in G} \left( |S_{v \narr r_v(y)}| + |S_{v \arr r_v(y)}^{\exp}| \right)\stackrel{*}{\leq} 2(k+1) \sum_{v \in G} |G^k_{v \narr r_v(y)}|\stackrel{**}{\leq} 2(k+1)^2 n~.  
	\end{align*}
	In $*$ we use that $S_{v \narr r_v(y)} \subseteq G^k_{v \narr r_v(y)}$ and that $|S_{v \arr r_v(y)}^{\exp}| \leq k \cdot |G^k_{v \narr r_v(y)}|$ by guarantee (ii) of \Cref{lem:mainProbing} . In $**$ we use \Cref{obs:globalmemory} to upper bound the sum. 
\end{proof}

\begin{proof}[Proof of \Cref{lem:balancedExp}: Global space]
	The global space bound we aim for, i.e., how much space do all machine collectively use, is $O(n \cdot \poly(k))$. Let $v \in G$. For every element $w \in S_v$, node $v$ stores a tuple $(v,r_v(w),w,r_w(v),d_G(v,w))$, which takes $5$ words. Additionally, node $v$ stores its state, and the state of all its directions, which are bounded by $\deg(v)$. Storing a state takes one word. Hence, the global space is 
	\begin{align*}
		\sum_{v \in G} \left( 5|S_v| + 1 + \deg(v) \right) &= n + 5 \sum_{v \in G} |S_v| + \sum_{v \in G} \deg(v) \\
		&\leq n + O(n \cdot \poly(k)) + 2m = O(n \cdot \poly(k))~,
	\end{align*}
	
	where the inequality comes from \Cref{lem:globalMemoryBound} and the Handshaking lemma. Note that we can bound $m$ by $n$, since in a tree, $n=m+1$.
\end{proof}

\subsection{Probing Procedure} \label{sec:probingProcedure}

Our probing procedure is an integral part of \Cref{alg:balancedExp}, as it steers the exponentiation of nodes in certain (safe) directions.

\lemMainProbing*

\begin{algorithm}[H]
	\caption{Probing}\label{alg:probing}
	\begin{algorithmic}[1]
		
		\Function{\pro}{$k,B_v$} 
		\State{For every neighbor $x \in N(v)$, compute
			$U_{v \arr x} = \sum_{w \in S_{v \arr x}} |S_{w \narr r_w(v)} \cap N^{k-d(v,w)}(w)|$} \phantomsection\label{l:computeU}
		\State{Define $\hds \coloneqq \{ x \in N(v) \mid U_{v \arr x} > k \cdot n^{\eps}\} \cup B_v$} \phantomsection\label{l:blockedDirs}
		\If{$\hds = \emptyset$}
		\State{Define $y \coloneqq \argmax_u\{U_{v\arr u}\}$} \phantomsection\label{l:ynode}
		\EndIf
		\State{Define $X \coloneqq N(v) \setminus  \{ \hds \cup \{y\} \}$}\phantomsection\label{l:yupdate}
		\State{\Return{$\hds,X$}}	
		\EndFunction
		
	\end{algorithmic}
\end{algorithm}

The next lemma relates the size of the memory of a node $v$ in direction $z\in N(v)$, that is, the size of $\expo(z,k)$ to the value of the probing $U_{v \arr z}$. The lower bound utilizes that we are in a tree. 
\begin{lemma} \label{lem:approxFactor}
	For any $v\in G$ and $z \in N(v)$ it holds that $U_{v \arr z} / k \leq |\expo(z,k)| \leq U_{v \arr z}$.
\end{lemma}

\begin{proof}
	Let us compare the definitions of $U_{v \arr z}$ and the size of $\expo(z,k)$:
	\begin{align*}
		U_{v \arr z} &= \sum_{w \in S_{v \arr z}} |S_{w \narr r_w(v)} \cap N^{k-d(v,w)}(w)| \\
		|\expo(z,k)| &= \left| \bigcup_{w \in S_{v \arr z}} S_{w \narr r_w(v)} \cap N^{k-d(v,w)}(w) \right|,
	\end{align*}
	
	where $N^{k-d(v,w)}(w)$ is the $k-d(v,w)$ radius subgraph centered at $w$. By comparing the two, the upper bound of  $|\expo(z,k)| \leq U_{v \arr z}$ clearly holds. 
	
	The lower bound is more subtle. Let us compute how many times a node $w$ in the sum of $U_{v \arr z}$ can be overcounted. Consider the unique path $P_{vw}$ from $v$ to a node $w$. Observe that out of all nodes in $S_{v \arr z}$, node $w$ is in set $S_{y \narr v}$ only for nodes $y \in P_{vw}$. Since $|P_{vw}| \leq k + 1$, any node $w$ is overcounted at most $k+1$ times, proving the lower bound of $U_{v \arr z} / k \leq |\expo(z,k)|$.
\end{proof}

\begin{corollary} \label{cor:boundbyGraph}
	For any $v\in G$ and $z\in N(v)$ we have $U_{v \arr z} \leq k \cdot |G^k_{v \arr z}|$.
\end{corollary}

\begin{proof}
	By the lower bound in \Cref{lem:approxFactor}, it holds that $U_{v \arr z}/k \leq |\expo(z,k)|$, and since $\expo(z,k) \subseteq G^k_{v \arr z}$ the claim follows.
\end{proof}

\begin{proof}[Proof of \Cref{lem:mainProbing} (i)]
	Observe that $\hds =\{ x \in N(v) \mid U_{v \arr x} > k \cdot n^{\eps}\} \cup B_v$. Since $\hds\cap X=\emptyset$, for every $x \in X$ it holds that $U_{v \arr x} \leq k \cdot n^{\eps}$. Hence, by the upper bound of \Cref{lem:approxFactor} it holds that $|\expo(z,k)| \leq U_{v \arr x} \leq k \cdot n^\eps$.     
\end{proof}

\begin{proof}[Proof of \Cref{lem:mainProbing} (ii)]
	
	Consider $b\in N(v)\setminus X$.
	\begin{itemize}
		\item If $b=y$ then by $y$'s definition, we obtain that $U_{v \arr x}\leq U_{v \arr b}$ holds for all $x\in X$. If $b \in \hds$ from \Cref{l:computeU}, it holds that  $U_{v \arr x}  \leq k \cdot n^{\eps} < U_{v \arr b}$. In both cases, we deduce for all $x\in X$ that
		\begin{align*}
			|\expo(x,k)|\stackrel{Lem.~\ref{lem:approxFactor}}{\leq} U_{v \arr x} \leq U_{v \arr b} \stackrel{Cor.~\ref{cor:boundbyGraph}}{\leq k} \cdot |G^k_{v \arr b}|\leq k \cdot |G^k_{v \narr x}|~,
		\end{align*}
		where the last inequality holds as we have $G^k_{v \arr b}\subseteq G^k_{v \narr x}$ (here we use $x\neq b$). 
		\item If $b \in \hds$ is originally from $B_v$ it holds that $|G^k_{v \arr b}|>n^\eps$ and 
		\begin{equation*}
			|\expo(x,k)|\stackrel{Lem.~\ref{lem:approxFactor}}{\leq}  U_{v \arr x}  \leq k \cdot n^{\eps} < k \cdot |G^k_{v \arr b}| \qedhere
		\end{equation*}
	\end{itemize}
\end{proof}

\begin{proof}[Proof of \Cref{lem:mainProbing} (iii)]
	Let $v \in G$ be any node. If $|\hds| = 0$, then node $y$ is defined and \Cref{l:yupdate} immediately implies $|X|= \deg(v)-1$. If $|\hds|=1$, then node $y$ is not defined and again, \Cref{l:yupdate} immediately implies $|X|= \deg(v)-1$.
	
	Assume for contradiction that $|\hds| \geq 2$ and that $v$ is an important node. There exists distinct neighbors $y_1,y_2 \in \hds$. If $v$ is an important node, for some neighbor $u$, it holds that $|G_{v \narr u}|\leq n^\eps$. So w.l.o.g. it holds that $G_{v \arr y_1} \subseteq G_{v \narr u}$. Because $y_1 \in \hds$, it holds that either $U_{v \arr y_1} > k \cdot n^{\eps}$ ($y_1$ from \Cref{l:computeU}), implying $|G^k_{v \arr y_1}| > n^{\eps}$ by \Cref{cor:boundbyGraph}, or $|G^k_{v \arr y_1}|>n^\eps$ ($y_1$ from $B_v$). Both cases are a contradiction. \qedhere
\end{proof}

\begin{proof}[Proof of \Cref{lem:mainProbing}: \mpc details]
	In \Cref{alg:probing}, we only need to ensure that every node $v \in G$ can compute $U_{v \arr x}$ for every neighbor $x \in N(v)$ in $O(1)$ time and $O(|S_v|)$ space. Observe that $U_{v \arr x}$ is only a simplified version of $\expo(x,k)$, and hence $U_{v \arr x}$ at most as hard to compute as $\expo(x,k)$. In \Cref{sec:MPCdetails} we argue how $\expo(x,k)$ (and hence how $U_{v \arr x}$) can be computed. There is one small caveat: in \Cref{sec:MPCdetails}, for \expo, we assume that nodes $w \in G$ with $\deg(w)>n^{3\eps}$ never get queried by a node $v \in G$ with $w \in S_{v \arr x}$ because \Cref{alg:probing} blocks directions $x$ for $v$. However, in \Cref{alg:probing} we still have to compute value $U_{v \arr x}$. 
	
	Turns out that in this case there is a shortcut, and we don't actually have to compute $U_{v \arr x}$ explicitly. If $\deg(w)>n^{3\eps}$ for a node $w \in S_{v \arr x}$ such that $v$ wants to query $w$, we simply set $U_{v \arr x}$ to be $n^{3\eps}$, which is clearly a lower bound for what $U_{v \arr x}$ should be. This lower bound is however large enough for $x$ to be added to \hds, since $n^{3\eps}=n^{3\delta/8} > n^{2\delta/8} \geq k \cdot n^\eps$. 
\end{proof}

\subsection{Missing \mpc Details} \label{sec:MPCdetails}

This section is dedicated to showing how \Cref{alg:balancedExp} can be implemented in the low-space \mpc model. So far we have taken a node-centric approach to \mpc, where we reason that a certain node $v$ can do something. In practice, it is always the machine(s) $M(v)$ storing $v$ that is performing these actions. Think of $M(v)$ as a function that returns the address of the machine that stores node $v$ and its incident edges. Observe that for two distinct nodes $v$ and $w$, it may be that $M(v)=M(w)$.

In the proof of \Cref{lem:balancedExp} (divided into separate sections) we have so far reasoned that the local space of a machine is bounded by $O(n^\delta)$, and that the total space is bounded by $O(n \cdot \poly(k))$. In \Cref{lem:validTuple,lem:learnInducedGraph,lem:learnLabel} we showed that the algorithm can be implemented as long as nodes can query for tuples (of constant word size) from other nodes. Hence, in order to complete the proof of \Cref{lem:balancedExp}, we just have to ensure that the aforementioned communication is feasible both ways (issuing and answering) in the low-space \mpc model. In particular, we need to ensure that communication bandwidth of $O(n^\delta)$ per machine is respected throughout the algorithm.

Initially, before executing \Cref{alg:balancedExp}, the input graph of $n$ nodes and $m$ edges is distributed among the machines arbitrarily. By applying \Cref{def:aggTreeStructure}, we can organize the input such that every node and it's edges are hosted on a single machine, or, in the case of high degree, on multiple consecutive machines.
For the following arguments, assume that both node $v$ and its set $S_v$ is stored on a single machine $M(v)$, or, in the case of high degree, on multiple consecutive machines. 

\subparagraph*{Issuing Queries.} Let $v \in G$. If during the algorithm it holds that $|S_v| > n^{3\eps}$, $v$ turns \full and does not issue queries for the remainder of the algorithm.
Hence, we can assume that if node $v$ is issuing a query, it holds that $|S_v| \leq n^{3\eps}$ and $S_v$ is fully contained in the memory of $M(v)$. Issuing a query to a node $w \in S_v$ consists of sending a message to $M(w)$, asking for a subset of tuples in $S_w$, e.g., $S_{w \narr r_w(v)} \cap N^{k-d(v,w)}(w)$. Since such a message is of constant word size (contains only direction $r_w(v)$ and value $k-d(v,w)$), machine $M(v)$ sends $O(|S_v|)$ messages for node $v$. Hence, machine $M(v)$ sends $O(n^\delta)$ messages in total (accounting for all nodes it stores). By the same logic, every machine $M(v)$ receives $O(n^\delta)$ messages in total. 

\subparagraph*{Answering Queries.} Let $w \in G$. This is more complex than issuing queries because nodes outside set of $S_w$ may be querying node $w$. Observe that the answer to a query for a node $w \in G$ is of size at most $|S_w|$. Observe that if $\deg(w) > n^{3\eps}$ for a node $w \in G$, no node $v$ will query $w$, because direction $r_v(w)$ will be \blocked by \Cref{lem:mainProbing}. Hence, if a node $w \in G$ is queried, the answer is contained in $S_w$ and is of size at most $|S_w|\leq n^{6\eps}$ by \Cref{lem:localMemoryImportant,lem:localMemoryNonImportant} (it fully fits into one machine).

Denote the collection of machines we are using for the algorithm as $\mathcal{M} = \{M_1,M_2,\dots,M_l\}$. We allocate a collection of machines $\mathcal{M}' = \{M'_1,M'_2,\dots,M'_{O(l)}\}$ for answering queries (we wipe them clean before every communication round). If a node $v$ wants to query a node $w \in S_v$ the machine in $\mathcal{M}$ storing $v$ sends a message $(w,v)$ to $\mathcal{M}'$. Additionally, for every set $S_w$, such that $|S_w| \leq n^{6\eps}$ (i.e., a set that could contain an answer to some query), that a machine in $\mathcal{M}$ stores, it sends message $(w,S_w)$ to $\mathcal{M}'$. It is important that message $(w,S_w)$ fully fits into one machine in $\mathcal{M}'$. We sort all messages in $\mathcal{M}'$ by the first entry of a message using \Cref{def:aggTreeStructure}. Now consider a machine $M'_j \in \mathcal{M}'$ storing message $(w,S_w)$. If $M'_j$ sees all queries directed at $w$ in its own or $M'_{j+1}$'s memory, it can answer the queries itself in $O(1)$ time. Otherwise, the set of machines $M'_j,M'_{j+1},\dots$ that store all queries for $w$ can form a broadcast tree such that $M'_j$ is the root. The root can then propagate $S_w$ to all machines in the broadcast tree in $O(1)$ time, which in turn can answer the queries they hold in memory.

\begin{definition}[Aggregation Tree Structure, \cite{BKM20}] \label{def:aggTreeStructure}
	Assume that an \mpc algorithm receives a collection of sets $A_1,\dots, A_k$ with elements from a totally ordered domain as input. In an aggregation tree structure for $A_1,\dots,A_k$, the elements of $A_1,\dots,A_k$ are stored in lexicographically sorted order (they are primarily sorted by the number $i \in \{1,\dots,k\}$ and within each set $A_i$ they are sorted increasingly). For each $i \in \{1,\dots,k\}$ such that the elements of $A_i$ appear on at least 2 different machines, there is a tree of constant depth containing the machines that store elements of $A_i$ as leafs and where each inner node of the tree has at most $n^{\delta/2}$ children. The tree is structured such that it can be used as a search tree for the elements in $A_i$ (i.e., such that an in-order traversal of the tree visits the leaves in sorted order). Each inner node of these trees is handled by a separate additional machine. In addition, there is a constant-depth aggregation tree of degree at most $n^{\delta/2}$ connecting all the machines that store elements of $A_1 ,\dots, A_k$.
\end{definition}

\subsection{Proof of Lemma \ref{lem:subsets}} \label{sec:proofOfLemma}

Let us make the following observation.

\begin{observation} \label{obs:importantInSubtree}
	If a node $v$ is in a subtree of size at most $n^{\delta/10}$, it is important. 
\end{observation}

\begin{proof}
	Let node $v$ be as in the observation statement. By \Cref{def:subtree}, $v$ has a neighbor $u$ for which it holds that $|G_{v \narr u}| \leq n^{\delta/10}$. Since important nodes are defined by \Cref{def:importantNode} as having a neighbor $u$ such that $|G_{v \narr u}| \leq n^{\delta/8}$ holds, the claim follows.
\end{proof}

\begin{proof}[Proof of \Cref{lem:subsets}]
	
	Consider applying \Cref{lem:balancedExp} with parameter $k= 100 \log n$. By the lemma statement, for every important node $v \in F$ there is a node $z$ and a machine that for all $x \in N(v) \setminus z$ holds $N^{ 100 \log n}(v) \cap F_{v \arr x}$ in memory. Let every such machine form a set 
	\begin{align*}
		U_v \coloneqq \bigcup_{x \in N(v) \setminus z} N^k(v) \cap F_{v \arr x}
	\end{align*}
	
	and observe that $U_v \subseteq S_v$ in \Cref{alg:balancedExp}.
	Define a collection $\mathcal{U} = \{ U_v \}$, which fulfills the requirements of \Cref{lem:subsets}:
	
	\begin{enumerate}
		\item $|U_v| \leq |S_v| < n^{3\eps} = n^{3\delta/8} < n^\delta$ for every set $U_v \in \mathcal{U}$ by \Cref{lem:localMemoryImportant}.
		\item $\sum_{\mathcal{U}} U_v \leq \sum_{v \in G} |S_v| = O(n \cdot \poly(\log n))$ by \Cref{lem:globalMemoryBound}.
		\item Every node $v \in F$ which is contained in a subtree of size at most $n^{\delta/10}$ is important by \Cref{obs:importantInSubtree}. By \Cref{lem:balancedExp}, there is a set $U_v \in \mathcal{U}$ for every important node $v$, and $U_v$ is a good subset for $v$ (by \Cref{def:good_subset}) because for $\deg(v)-1$ distinct directions $x$ it holds that $N^{ 100 \log n}(v) \cap F_{v \arr x} = F^{100 \log n}_{v \arr x} \subseteq U_v$.
	\end{enumerate}
	
	Moreover, the machine storing node $v$ and its set $U_v$ can compute $F[U_v]$ in $O(1)$ rounds by \Cref{lem:learnInducedGraph}. Because in a forest it holds that there are more nodes than edges, the machine storing node $v$ and its set $U_v$ can also store $F[U_v]$.
	
	Using \Cref{lem:balancedExp} with parameter $k= 100 \log n$ takes $O(\log \log n)$ low-space \mpc rounds and $O(n \cdot \poly(\log n))$ global space, completing the proof.
\end{proof}

\newpage
\bibliographystyle{alpha}
\bibliography{mpc-coloring-arxiv}

\newcommand{\etalchar}[1]{$^{#1}$}
\begin{thebibliography}{GGK{\etalchar{+}}18}

\bibitem[ABI86]{alon86}
Noga Alon, L\'{a}sl\'{o} Babai, and Alon Itai.
\newblock {A Fast and Simple Randomized Parallel Algorithm for the Maximal
  Independent Set Problem}.
\newblock {\em Journal of Algorithms}, 7(4):567--583, 1986.

\bibitem[BBD{\etalchar{+}}19]{behnezhadMIS}
Soheil Behnezhad, Sebastian Brandt, Mahsa Derakhshan, Manuela Fischer,
  MohammadTaghi Hajiaghayi, Richard~M. Karp, and Jara Uitto.
\newblock Massively parallel computation of matching and mis in sparse graphs.
\newblock In {\em Proceedings of the 2019 ACM Symposium on Principles of
  Distributed Computing}, PODC '19, page 481–490, New York, NY, USA, 2019.
  Association for Computing Machinery.

\bibitem[BCM{\etalchar{+}}21]{BCMOS21}
Alkida Balliu, Keren Censor{-}Hillel, Yannic Maus, Dennis Olivetti, and Jukka
  Suomela.
\newblock {Locally Checkable Labelings with Small Messages}.
\newblock In {\em the Proceedings of the International Symposium on Distributed
  Computing (DISC)}, pages 8:1--8:18, 2021.

\bibitem[BDE{\etalchar{+}}19]{Behnezhad2019}
Soheil Behnezhad, Laxman Dhulipala, Hossein Esfandiari, Jakub Łącki, and
  Vahab Mirrokni.
\newblock {Near-Optimal Massively Parallel Graph Connectivity}.
\newblock In {\em FOCS}, 2019.

\bibitem[BE10]{BE10}
Leonid Barenboim and Michael Elkin.
\newblock Sublogarithmic distributed {MIS} algorithm for sparse graphs using
  nash-williams decomposition.
\newblock {\em Distributed Comput.}, 22(5-6):363--379, 2010.

\bibitem[BEPS16]{BEPSv3}
Leonid Barenboim, Michael Elkin, Seth Pettie, and Johannes Schneider.
\newblock {The Locality of Distributed Symmetry Breaking}.
\newblock {\em Journal of the ACM}, 63(3):20:1--20:45, 2016.

\bibitem[BFU21]{BRANDT202122}
Sebastian Brandt, Manuela Fischer, and Jara Uitto.
\newblock {Breaking the Linear-memory Barrier in MPC: Fast MIS on Trees with
  Strongly Sublinear Memory}.
\newblock {\em Theoretical Computer Science}, 849:22--34, 2021.

\bibitem[BKM20]{BKM20}
Philipp Bamberger, Fabian Kuhn, and Yannic Maus.
\newblock {Efficient Deterministic Distributed Coloring with Small Bandwidth}.
\newblock In {\em {PODC} '20: {ACM} Symposium on Principles of Distributed
  Computing (PODC)}, pages 243--252, 2020.

\bibitem[BLM{\etalchar{+}}23]{sodapaper}
Alkida Balliu, Rustam Latypov, Yannic Maus, Dennis Olivetti, and Jara Uitto.
\newblock {Optimal Deterministic Massively Parallel Connectivity on Forests}.
\newblock In {\em Proceedings of the 2023 Annual ACM-SIAM Symposium on Discrete
  Algorithms (SODA)}, pages 2589--2631, 2023.

\bibitem[CC22]{ccderandom}
Sam Coy and Artur Czumaj.
\newblock {Deterministic Massively Parallel Connectivity}.
\newblock In {\em Proceedings of the ACM Symposium on Theory of Computing
  (STOC)}, 2022.

\bibitem[CDP20]{Czumaj2020}
Artur Czumaj, Peter Davies, and Merav Parter.
\newblock {Graph Sparsification for Derandomizing Massively Parallel
  Computation with Low Space}.
\newblock In {\em the Proceedings of the Symposium on Parallelism in Algorithms
  and Architectures (SPAA)}, pages 175--185, 2020.

\bibitem[CDP21a]{componentstable}
Artur Czumaj, Peter Davies, and Merav Parter.
\newblock {Component Stability in Low-Space Massively Parallel Computation}.
\newblock In {\em PODC}, 2021.

\bibitem[CDP21b]{Czumaj2021}
Artur Czumaj, Peter Davies, and Merav Parter.
\newblock {Improved Deterministic $(\Delta+1)$-Coloring in Low-Space MPC}.
\newblock In {\em PODC}, pages 469–--479, 2021.

\bibitem[CDP21c]{CDP21}
Artur Czumaj, Peter Davies, and Merav Parter.
\newblock Simple, deterministic, constant-round coloring in congested clique
  and {MPC}.
\newblock {\em {SIAM} Journal on Computing}, 50(5):1603--1626, 2021.

\bibitem[CDP21d]{Czumaj-congested-coloring}
Artur Czumaj, Peter Davies, and Merav Parter.
\newblock {Simple, Deterministic, Constant-Round Coloring in Congested Clique
  and MPC}.
\newblock {\em SIAM Journal on Computing}, 50(5):1603--1626, 2021.

\bibitem[CFG{\etalchar{+}}19]{Chang2019}
Yi-Jun Chang, Manuela Fischer, Mohsen Ghaffari, Jara Uitto, and Yufan Zheng.
\newblock {The Complexity of {$(\Delta+1)$}-Coloring in Congested Clique,
  Massively Parallel Computation, and Centralized Local Computation}.
\newblock In {\em PODC}, 2019.

\bibitem[Cha20]{Chang2020}
Yi-Jun Chang.
\newblock {The Complexity Landscape of Distributed Locally Checkable Problems
  on Trees}.
\newblock In {\em DISC}, pages 18:1--18:17, 2020.

\bibitem[CP19]{CP19}
Yi{-}Jun Chang and Seth Pettie.
\newblock {A Time Hierarchy Theorem for the LOCAL Model}.
\newblock {\em {SIAM} J. Comput.}, 48(1):33--69, 2019.

\bibitem[DG08]{Dean2008}
Jeffrey Dean and Sanjay Ghemawat.
\newblock {MapReduce: Simplified Data Processing on Large Clusters}.
\newblock {\em Communications of the ACM}, pages 107--113, 2008.

\bibitem[GFG23]{GFG23}
Jeff Giliberti, Manuela Fischer, and Christoph Grunau.
\newblock Deterministic massively parallel symmetry breaking for sparse graphs.
\newblock {\em CoRR}, abs/2301.11205, 2023.

\bibitem[GGH{\etalchar{+}}23]{Ghaffari2023-soda}
Mohsen Ghaffari, Christoph Grunau, Bernhard Haeupler, Saeed Ilchi, and Václav
  Rozhoň.
\newblock {Improved Distributed Network Decomposition, Hitting Sets, and
  Spanners, via Derandomization}.
\newblock In {\em Proceedings of the 2023 Annual ACM-SIAM Symposium on Discrete
  Algorithms (SODA)}, pages 2532--2566, 2023.

\bibitem[GGJ20]{GGJ20}
Mohsen Ghaffari, Christoph Grunau, and Ce~Jin.
\newblock {Improved MPC Algorithms for MIS, Matching, and Coloring on Trees and
  Beyond}.
\newblock In {\em DISC}, 2020.

\bibitem[GGK{\etalchar{+}}18]{linear-MIS}
Mohsen Ghaffari, Themis Gouleakis, Christian Konrad, Slobodan Mitrovic, and
  Ronitt Rubinfeld.
\newblock {Improved Massively Parallel Computation Algorithms for MIS,
  Matching, and Vertex Cover}.
\newblock In {\em PODC}, pages 129--138, 2018.

\bibitem[Gha16]{GhaffariImproved16}
Mohsen Ghaffari.
\newblock {An Improved Distributed Algorithm for Maximal Independent Set}.
\newblock In {\em Proceedings of the 2023 Annual ACM-SIAM Symposium on Discrete
  Algorithms (SODA)}, pages 270--277, 2016.

\bibitem[GKU19]{Ghaffari2019}
Mohsen Ghaffari, Fabian Kuhn, and Jara Uitto.
\newblock {Conditional Hardness Results for Massively Parallel Computation from
  Distributed Lower Bounds}.
\newblock In {\em FOCS}, pages 1650--1663, 2019.

\bibitem[GLM{\etalchar{+}}23]{gupta2023fast}
Chetan Gupta, Rustam Latypov, Yannic Maus, Shreyas Pai, Simo Särkkä, Jan
  Studený, Jukka Suomela, Jara Uitto, and Hossein Vahidi.
\newblock Fast dynamic programming in trees in the mpc model, 2023.

\bibitem[GS19]{Ghaffari2019-arboricity}
Mohsen Ghaffari and Ali Sayyadi.
\newblock {Distributed Arboricity-Dependent Graph Coloring via All-to-All
  Communication}.
\newblock In {\em ICALP}, pages 142:1--142:14, 2019.

\bibitem[GSZ11]{goodrich}
Michael~T. Goodrich, Nodari Sitchinava, and Qin Zhang.
\newblock Sorting, searching, and simulation in the mapreduce framework.
\newblock In Takao Asano, Shin-ichi Nakano, Yoshio Okamoto, and Osamu Watanabe,
  editors, {\em Algorithms and Computation}, pages 374--383, Berlin,
  Heidelberg, 2011. Springer Berlin Heidelberg.

\bibitem[GU19]{GU19}
Mohsen Ghaffari and Jara Uitto.
\newblock {Sparsifying Distributed Algorithms with Ramifications in Massively
  Parallel Computation and Centralized Local Computation}.
\newblock In {\em SODA}, 2019.

\bibitem[IBY{\etalchar{+}}07]{Isard2007}
Michael Isard, Mihai Budiu, Yuan Yu, Andrew Birrell, and Dennis Fetterly.
\newblock {Dryad: Distributed Data-Parallel Programs from Sequential Building
  Blocks}.
\newblock {\em ACM SIGOPS Operating Systems Review}, pages 59--72, 2007.

\bibitem[KSV10]{KarloffSV10}
Howard~J. Karloff, Siddharth Suri, and Sergei Vassilvitskii.
\newblock {A Model of Computation for MapReduce}.
\newblock In {\em SODA}, 2010.

\bibitem[Lin87]{linial}
Nathan Linial.
\newblock {Distributive Graph Algorithms -- Global Solutions from Local Data}.
\newblock In {\em FOCS}, 1987.

\bibitem[Lin92]{Linial92}
Nathan Linial.
\newblock {Locality in Distributed Graph Algorithms}.
\newblock {\em {SIAM} J. Comput.}, 21(1):193--201, 1992.

\bibitem[LMSV11]{filtering2011}
Silvio Lattanzi, Benjamin Moseley, Siddharth Suri, and Sergei Vassilvitskii.
\newblock {Filtering: A Method for Solving Graph Problems in MapReduce}.
\newblock In {\em SPAA}, pages 85--94, 2011.

\bibitem[LU21]{LU21}
Rustam Latypov and Jara Uitto.
\newblock Deterministic 3-coloring of trees in the sublinear {MPC} model.
\newblock {\em CoRR}, abs/2105.13980, 2021.

\bibitem[Lub86]{luby86}
Michael Luby.
\newblock {A Simple Parallel Algorithm for the Maximal Independent Set
  Problem}.
\newblock {\em SIAM Journal on Computing}, 15:1036--1053, 1986.

\bibitem[LW10]{wattenhofer}
Christoph Lenzen and Roger Wattenhofer.
\newblock {Brief Announcement: Exponential Speed-Up of Local Algorithms Using
  Non-Local Communication}.
\newblock In {\em PODC}, 2010.

\bibitem[NW64]{NashWilliams1964}
Crispin Nash-Williams.
\newblock {Decomposition of Finite Graphs Into Forests}.
\newblock {\em Journal of the London Mathematical Society}, s1-39:12, 1964.

\bibitem[OJ99]{coloring-simple}
\"{O}jvind Johansson.
\newblock {Simple Distributed $\Delta + 1$-coloring of Graphs}.
\newblock {\em Information Processing Letters}, pages 229--232, 1999.

\bibitem[Par18]{P18}
Merav Parter.
\newblock $(\delta+1)$ coloring in the congested clique model.
\newblock In Ioannis Chatzigiannakis, Christos Kaklamanis, D{\'{a}}niel Marx,
  and Donald Sannella, editors, {\em {International Colloquium on Automata,
  Languages, and Programming, (ICALP)}}, volume 107, pages 160:1--160:14, 2018.

\bibitem[PS18]{PS18}
Merav Parter and Hsin{-}Hao Su.
\newblock Randomized $(\delta+1)$-coloring in $o(\log^* \delta)$ congested
  clique rounds.
\newblock In {\em 32nd International Symposium on Distributed Computing
  (DISC)}, volume 121, pages 39:1--39:18, 2018.

\bibitem[RG20]{RG20}
V{\'{a}}clav Rozho\v{n} and Mohsen Ghaffari.
\newblock {Polylogarithmic-time deterministic network decomposition and
  distributed derandomization}.
\newblock In {\em STOC}, pages 350--363, 2020.

\bibitem[RVW18]{Roughgarden18}
Tim Roughgarden, Sergei Vassilvitskii, and Joshua~R. Wang.
\newblock {Shuffles and Circuits (On Lower Bounds for Modern Parallel
  Computation)}.
\newblock {\em Journal of the ACM}, 2018.

\bibitem[Whi09]{White2009}
Tom White.
\newblock {\em {Hadoop: The Definitive Guide}}.
\newblock O'Reilly Media, Inc., 2009.

\bibitem[ZCF{\etalchar{+}}10]{Zaharia2010}
Matei Zaharia, Mosharaf Chowdhury, Michael~J. Franklin, Scott Shenker, and Ion
  Stoica.
\newblock {Spark: Cluster Computing with Working Sets}.
\newblock In {\em the Proceedings of the SENIX Conference on Hot Topics in
  Cloud Computing (HotCloud)}, 2010.

\end{thebibliography}

\end{document}